\documentclass[letterpaper,twocolumn,aps,prl,superscriptaddress,floatfix]{revtex4-2}
\usepackage{float}
\usepackage{graphicx}
\usepackage{dcolumn}
\usepackage{bm}

\usepackage{ragged2e}
\usepackage{listings}
\usepackage{graphicx} 
\usepackage{mathtools}
\usepackage{enumitem}
\usepackage{amsfonts} 
\usepackage{indentfirst}
\usepackage[pdfstartview=XYZ,
bookmarks=true,
colorlinks=true,
linkcolor=blue,
urlcolor=blue,
citecolor=blue,
bookmarks=true,
linktocpage=true, 
hyperindex=true
]{hyperref}
\usepackage{orcidlink}
\usepackage{comment} 
\usepackage{physics}
\usepackage{caption}
\usepackage{subcaption}

\usepackage{amsmath,bm}
\usepackage{amssymb}
\usepackage{upgreek}
\usepackage{xcolor}
\usepackage{color,soul}
\usepackage[utf8]{inputenc}
\usepackage{amsthm}
\usepackage[ruled,vlined]{algorithm2e} 
\usepackage{dsfont}
\newcommand{\m}{\textemdash}

\newcommand{\A}{\mathcal{A}}
\newcommand{\F}{\mathcal{F}}

\newcommand{\X}{\mathcal{X}}
\newcommand{\Ss}{\mathcal{S}}
\newcommand{\argmax}{\operatorname*{argmax}}

\newcommand{\E}{\mathbb{E}}

\usepackage{caption}
\def\gbm#1{{\let\pi\uppi \let\phi\upphi \let\lambda\uplambda \let\mu\upmu \let\rho\uprho \let\sigma\upsigma \let\tau\uptau \let\theta\uptheta \let\eta\upeta \bm{#1}}}
\newcommand{\RC}[1]{\textcolor{blue}{[RC: #1]}}
 \newcommand{\hquad}{\hspace{0.5em}}

 \usepackage[english]{babel}

\newtheorem{theorem}{Theorem}
\newcommand{\D}{\text{D}}
\newtheorem{lemma}{Lemma}
\newtheorem{prop}{Proposition}
\newtheorem{corollary}{Corollary}[theorem]

\begin{document}

\title{Time-ordered free energy in correlated quantum systems: An agentic approach}

\author{Ruo Cheng Huang\,\orcidlink{0000-0001-8135-8693} }
\email{ruocheng.huang@ntu.edu.sg}
\affiliation{Nanyang Quantum Hub, School of Physical and Mathematical Sciences, Nanyang Technological University, Singapore
}
\affiliation{Centre for Quantum Technologies, Nanyang Technological University, Singapore}

\author{Isha Singh Le Xue}
\email{isha0005@e.ntu.edu.sg}
\affiliation{Nanyang Quantum Hub, School of Physical and Mathematical Sciences, Nanyang Technological University, Singapore
}
\affiliation{Centre for Quantum Technologies, Nanyang Technological University, Singapore}

\author{Yuxuan Qu}
 \email{quyu0001@e.ntu.edu.sg}
\affiliation{Nanyang Quantum Hub, School of Physical and Mathematical Sciences, Nanyang Technological University, Singapore
}
\affiliation{Centre for Quantum Technologies, Nanyang Technological University, Singapore}

\author{Paul M.\ Riechers\,\orcidlink{0000-0002-0135-3778}}
\email{pmriechers@gmail.com}
\affiliation{Beyond Institute for Theoretical Science (BITS), San Francisco, CA, USA}

\author{Varun Narasimhachar}
\email{varun.achar@gmail.com}
\affiliation{Institute of Advanced Intelligence and Computing (IAIC), Agency for Science, Technology and Research, 1 Fusionopolis Way, Republic of Singapore 138632}

\author{Mile Gu\,\orcidlink{0000-0002-5459-4313}}
\email{gumile@ntu.edu.sg}
\affiliation{Nanyang Quantum Hub, School of Physical and Mathematical Sciences, Nanyang Technological University, Singapore
}
\affiliation{Centre for Quantum Technologies, Nanyang Technological University, Singapore}
\affiliation{MajuLab, CNRS-UNS-NUS-NTU International Joint Research Unit, UMI 3654, 117543, Singapore}

\date{\today}

\begin{abstract}
How much work can an agent extract from a temporal sequence of quantum states when it can only operate online under causal constraints---deciding which energy extraction method to use with knowledge of what it has observed before? Here, we study this problem in the context of quantum state sequences that are potentially non-Markovian---generated by some underlying hidden Markov machine that the agent cannot observe. Using techniques from dynamic programming and computational mechanics, we present a method to identify the provably optimal agent strategy, with time complexity that scales linearly with sequence length. This motivates us to introduce the maximum work such agents can extract---\emph{time-ordered free energy} (TOFE)---as a fundamental measure of free energy available in a temporally correlated quantum system subject to causal considerations. 
\end{abstract}

\maketitle



\noindent \textit{Introduction} \m 
While free energy is typically discussed in the context of a state out of equilibrium at a fixed time, it also arises in temporal sequences of states. In such settings, extracting this free energy as useful work typically requires an agent that remembers past states and adapts its future actions accordingly. Consider an agent foraging in an environment that emits a sequence of bits alternating between $0$ and $1$. Each system appears maximally mixed to an agent with no memory and thus contains no free energy. However, an agent that tracks the previous bit can extract up to $k_\mathrm{B} T \ln 2$ of free energy at ambient temperature $T$~\cite{szilard1964decrease}. Classically, these ideas motivated the development of information ratchets. Operationally, they model an organism that is foraging for food or an edge device seeking a power source~\cite{still2012thermodynamics, kolchinsky2025maximizing,lloyd2025thermodynamics+}. Conceptually, they underpin our understanding of free energy in temporal correlations and the thermodynamics of learning and prediction~\cite{mandal2012work,boyd2016identifying,boyd2017leveraging,garner2017thermodynamics,boyd2018thermodynamics,garner2021fundamental,boyd2025thermodynamic}.

How do such ideas extend to quantum environments? Here, we consider an agent with access to a stream of $L$ quantum systems. The systems are temporally correlated in that they are generated by a hidden Markov model (HMM) whose states are hidden from the agent. The theoretically optimal strategy is to retain all systems and interact coherently with them at once~\cite{allahverdyan2004maximal,aaberg2013truly,brandao2013resource,horodecki2013fundamental,skrzypczyk2014work,Safranek23_Work}. However, such strategies require quantum memory that scales without bound with the length of the data stream. This conflicts with the perspective of an autonomous machine harvesting energy online and presents an immense practical challenge with current technology~\cite{lvovsky2009optical,heshami2016quantum}.

Consequently, we consider online agents bound by an \emph{arrow of time}: they must process the sequence of quantum states and extract energy sequentially, with actions at time $t$ depending only on past observations, thereby forbidding postponement of extraction until the end of the protocol. We ask \emph{how much energy an agent can harvest under such causality constraints and how this maximum is achieved}. This leads us to the \emph{time-ordered free energy} (TOFE) of a sequence of quantum states, defined as the maximum free energy available to a causally constrained agent. 

The TOFE may at first appear intractable to compute. Unlike classical bits, which can only be measured in a single basis, quantum systems allow many ways to harvest energy. Thus, an agent's possible action sequences scale exponentially with $L$, even when limited to finitely many actions per time step. Remarkably, we develop a dynamic programming (DP) algorithm that optimizes agent actions in time linear in $L$ for any fixed set of actions. We use this to answer a previously open question in sequential energy harvesting from quantum systems~\cite{huang2023engines}, showing that the best strategy is not necessarily to maximize the energy gain at each time step.

\noindent\textit{Framework}\m
Consider a source that emits a quantum system $Q_t$ in state $\sigma_{Q_{t}}^{(x_{t})}$ at the discrete time step $t$, where each state is identified by a label $x \in \mathcal{X}$. We model this source as a hidden Markov model (HMM), denoted by $\mathcal{M}=(\mathcal{S},(\sigma^{(x)})_{x\in\X},(\mathsf{T}^{(x)})_{x\in\X},\mu_{0})$, where $\mathcal{S}$ represents the set of latent states and $\mu_0$ defines their initial distribution. The dynamics of the system are governed by the elements of the transition matrix, $\mathsf{T}^{(x)}_{s,s'} = \Pr(S_t = s', X_t = x \mid S_{t-1} = s)$, which dictate the probability of transitioning from a latent state $s$ to $s'$ while simultaneously emitting the quantum state $\sigma^{(x)}$. Over a sequence of $L$ time steps, the process emits the multi-time quantum state $\rho^{(1:L)}$, expressed as:
\begin{equation}
\label{eq:fuel_state}
\rho^{(1:L)} = \sum_{x_{1:L}\in \X^L} \Pr(x_{1:L}) \bigotimes_{t=1}^L \sigma_{Q_t}^{(x_t)}~.
\end{equation}
A simple example of an HMM source, known as the \emph{perturbed coin}, can be found in Fig.~\ref{fig:combined_fig}(a). 

\begin{figure}[htbp]
\includegraphics[width=0.75\columnwidth]{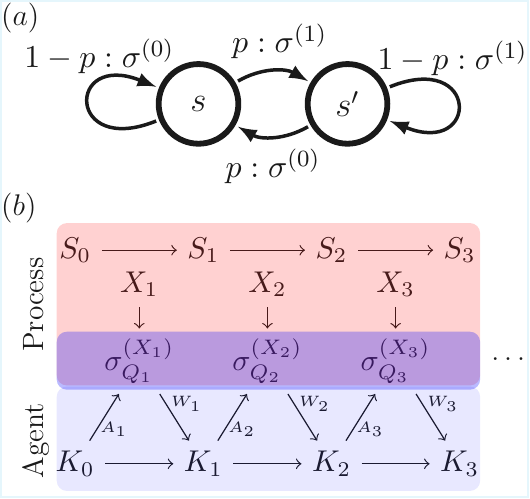}
    \caption{\justifying (a) The perturbed coin process is an example of a quantum hidden state source. It possesses two latent states $s$ and $s'$ (nodes). The directed edges describe its dynamics: an edge from $s$ to $s'$ labeled $p:\sigma$ indicates that a source in state $s$ has probability $p$ to transition to $s'$ and emit state $\sigma$. (b) A causal network depicting the interaction between the hidden quantum process (top, red) and the sequential agent (bottom, blue). The agent does not have access to $s_t$ and thus can only select an extraction method $A_t$ based on its current memory state $K_{t-1}$. The extracted work $W_t$ acts as an observable, allowing the agent to update its memory to $K_t$.}
\label{fig:combined_fig}
\end{figure}

We take on the role of the agent, who is given access to $Q_t$ at each time step $t$. We assume the agent lacks a quantum memory that persists between time steps, so it is restricted to performing local operations sequentially. It must obey \textit{causality}: its actions at any time step $t$ can depend only on actions and observations at times $t' < t$. \textit{For an agent constrained by locality and causality, the maximum expected extractable work} defines the TOFE, $\F_{\text{TO}}^{(L)}$.

We adopt the standard resource-theoretic framework for thermodynamics, in which an agent extracts work from a sequence of quantum systems $Q_t$ and stores it in an ideal battery $B_t$. Modeled as a semiclassical weight, the battery ensures that the extracted \emph{work} is defined as the change in its energy. At each time step $t$, the agent is restricted to thermal operations mediated by a memoryless thermal reservoir $R$, which is initialized in the thermal state $\gamma=Z^{-1}e^{-\beta\mathcal{H}_{R}}$ with partition function $Z=\tr(e^{-\beta\mathcal{H}_{R}})$ at inverse temperature $\beta=(k_BT)^{-1}$. These widely shared assumptions establish fundamental thermodynamic benchmarks~\cite{kosloff2013quantum,brandao2013resource,skrzypczyk2014work,aaberg2014catalytic,gour2015resource,lostaglio2015description,korzekwa2016extraction,lostaglio2017thermodynamic,sparaciari2017resource}. We define the action space as $\mathcal{A}=\{\mathcal{W}_{\rho}\}$, parametrized by target states $\rho$, where $\mathcal{W}_{\rho}$ is tailored to extract all non-equilibrium free energy from $\rho$~\cite{riechers2021initial,huang2023engines,lumbreras2025quantum}. Here, we assume that the Hamiltonian governing $Q_t$ is energy-degenerate, and revisit the non-degenerate case in Appendix~\ref{sec:nondegen}.

The most general means for an agent to harvest work sequentially is via a \emph{causal strategy} $\Gamma$, a systematic rule for choosing which action $a_t \in \mathcal{A}$ to apply based on the prior history $\mathbf{h}_{t-1}=(a_1,w_1,\ldots,a_{t-1},w_{t-1})$. The TOFE over $L$ time steps is then the maximum extracted work when optimized over all causal strategies:
\begin{align}
\label{eq:max_expected_work}
    \F^{(1:L)}_{\text{TO}}(\rho^{(1:L)})\!\coloneqq\!\max_{\Gamma}\E_{\mathbf{H}_L|\Gamma}\left[\sum_{t=1}^LW_t\right]~,
\end{align}
where the expectation is taken over $\mathbf{H}_L|\Gamma$, the random variable governing $\mathbf{h}_L$ when the agent executes $\Gamma$.

Na\"\i vely, executing such an optimal strategy could require the agent to store the entire history $\mathbf{h}_t$. In Appendix \ref{sec:sufficient_statistic}, we show that this is unnecessary. Instead, an agent can store a suitable \emph{belief state} 
\begin{equation} \gbm\eta_t \coloneqq \bigl[ \Pr(S_t=s | \mathbf{h}_{t}) \bigr]_{s\in \mathcal{S}}~, \end{equation}
representing its \emph{subjective} belief of which latent state the source is in, conditioned on its action-observation history. After each round, the agent
updates this belief based on the observed work value using Bayesian inference~\cite{cox1946probability,jaynes2003probability}. Specifically, upon observing a work value $w_t$, the agent's belief state is updated according to the Bayesian update rule $\tau$,
\begin{equation}
\begin{split}
\label{eq:bayesian_update}
\gbm\eta_{t} &= \tau(\gbm\eta_{t-1}, a_t, w_t) \\
&\propto \sum_{x_{t}\in\mathcal{X}} \Pr(W_{t}=w_{t}|x_{t},a_{t},\gbm\eta_{t-1}) \gbm\eta_{t-1} \mathsf{T}^{(x_t)},
\end{split}
\end{equation}
where proportionality ensures normalization~\cite{crutchfield2016exact,huang2023engines} and $\mathsf{T}^{(x_t)}$ is the labeled transition matrix of the HMM. In Appendix~\ref{sec:mem_cost}, we show that updates can always be done near-reversibly, resulting in a negligible contribution to heat dissipation. 

Consequently, the history-dependent policy $\Gamma$ can be replaced by a belief-dependent policy $\Lambda$ that maps any belief state to an allowed action~\cite{aastrom1965optimal,smallwood1973optimal}. In the following, we will refer to $\Lambda$ simply as the \emph{policy}. At time $t$, an agent in the memory state $\gbm\eta_{t-1}$ executes an action $a_t=\Lambda(\gbm\eta_{t-1})\in\A$ to extract free energy from $\sigma_{Q_t}$, storing it in the battery $B_t$. The extracted work $w_t$ corresponds to the measurement result for $B_t$ in the energy eigenbasis~\cite{wigner1952messung,araki1960measurement}. Following this, the agent updates its memory to $\gbm\eta_t=\tau(\gbm\eta_{t-1},a_t,w_t)$. This protocol is illustrated in Fig.~\ref{fig:combined_fig}(b), where the random variables $S_t, X_t, A_t, K_t, W_t$ take the realizations $s_t, x_t, a_t, \gbm\eta_t, w_t$, respectively. Consequently, determining the TOFE and its optimal harvesting strategy reduces to finding the optimal policy $\Lambda$.


\noindent\textit{Policy optimization}\m Recall that the average work extracted by a protocol designed for a state $\rho$ when acting on a state $\sigma$ is given by~\cite{riechers2021initial}:
\begin{equation}
\label{eq:expected_reward_ub}
\beta\langle W_{\text{ext}}\rangle = \bigl[ \D(\sigma_Q\|\gamma_Q)-\D(\sigma_Q\|\rho) \bigr]~.
\end{equation}
Here, $\gamma_Q = Z_Q^{-1}e^{-\beta\mathcal{H}_Q}$ is the thermal state of $Q$ and $\D(\sigma\|\rho) = \tr(\sigma\log\sigma)-\tr(\sigma\log\rho)$ is the quantum relative entropy. The first term in Eq.~\eqref{eq:expected_reward_ub} represents the non-equilibrium free energy available in $\sigma$, while the second represents the dissipated heat due to expectation mismatch that vanishes if $\sigma=\rho$. Since our agent lacks knowledge of $\sigma$, the expected work extracted by taking action $a$ when the agent is in a belief state $\gbm\eta$, $\langle W(\gbm \eta,a)\rangle$, can be expressed as
\begin{equation}
\label{eq:final_rewards}
   \beta \langle W(\gbm \eta,a)\rangle \!\coloneq \D(\xi_{\gbm\eta}\|\gamma_Q) - \D(\xi_{\gbm\eta}\|\rho_a)~,
\end{equation}
where $\xi_{\gbm\eta}\coloneqq \sum_{x\in\X} \Pr(x|\gbm\eta)\sigma^{(x)}$ is the agent's expected quantum state, and $\rho_a$ is the target state of action $a$.

A possible policy is the ``greedy" or local-optimizing (LO) policy, which would choose an action that maximizes the immediate work extracted. Based on Eq.~\eqref{eq:final_rewards}, this LO policy would be tailored to the current expected state, $\rho_a=\xi_{\gbm\eta}$, thereby eliminating the dissipative term and maximizing the immediate expected work~\cite{huang2023engines}. However, strategically sacrificing immediate work by choosing a ``mismatched" protocol may yield greater predictive information for subsequent steps, leading to higher cumulative work extraction. The dissipation term $\D(\xi_{\gbm\eta}\|\rho_a)$ can thus be interpreted as the price of information—a trade-off between maximizing immediate reward and gaining knowledge for future benefit. 

Optimizing over all histories scales exponentially with the sequence length $L$. Instead, treating belief states $K_t$ as sufficient statistics enables recursive DP. For a fixed action set, the optimal value functions obtained from progressively finer discretizations of the continuous belief space converge to the optimal value function of the underlying continuous-state problem~\cite{lovejoy1991computationally,hauskrecht2000value}. We therefore formulate the resulting finite-state problem as a dynamic program, whose time complexity scales linearly with $L$.

\begin{theorem} 
For a given set of belief states $\mathcal{K}$ and an action set $\A$, there exists a dynamic programming algorithm to determine the optimal classical-causal policy $\Lambda^*$ in time $\mathcal{O}(L)$ that maximizes the average cumulative work extracted over $L$ time steps. Furthermore, if the agent's initial belief matches the initial distribution of the underlying HMM, this maximized cumulative work is mathematically equivalent to the TOFE, $\mathcal{F}_{\text{TO}}^{(1:L)}$ in Eq.~\eqref{eq:max_expected_work}.
\end{theorem}
We detail the algorithm (Alg.~\ref{alg:DPP_for_work}) and prove its optimality in Appendix~\ref{App:optimality_DP}, alongside the formal equivalence derivation in Appendix~\ref{App:equivalence}.
To ensure computational tractability at each step $t$, Theorem~\ref{thm2} (Appendix~\ref{sec:narrow_search}) proves it is sufficient to search only over the set of possible eigenbases rather than the full space of density matrices.

\noindent\textit{Thermodynamic hierarchy}\m
In the asymptotic limit, the TOFE rate is defined using Eq.~\eqref{eq:max_expected_work} as
\begin{equation}
\label{eq:asymp_extraction_rate}
f_{\text{TO}}\coloneqq\lim_{L\to\infty}\frac{1}{L} \F^{(1:L)}_{\text{TO}}(\rho^{(1:L)})~.
\end{equation}   
Here, transient ``terminal effects'' become negligible, and the optimal policy $\Lambda^*$ becomes stationary~\cite{sutton1998reinforcement,scherrer2012use} (see Appendix~\ref{sec:transient} for a discussion of the non-stationary regime), allowing us to approximate $f_{\text{TO}}$ using finite-horizon computation.

The stationary regime of $\Lambda^*$ induces a finite-state Markov chain on the agent's belief states. The TOFE rate can be found using the resulting limiting distribution.
\begin{equation}
    \beta f_{\text{TO}} \!=\!\E_{\gbm\eta\in\gbm \pi_{\lim}} \!\left[\D(\xi_{\gbm\eta}\|\gamma_Q)\!-\!\D(\xi_{\gbm\eta}\|\rho_{\Lambda^*(\gbm\eta)})\!\right].
\end{equation}
Here, the expectation is taken over all beliefs $\gbm\eta$ weighted by $\gbm\pi_{\text{lim}}$ and $\rho_{\Lambda^*(\gbm\eta)}$ is the target state selected by the optimal policy conditioned on belief state $\gbm\eta$.  

We contextualize TOFE against two benchmarks. For any multi-time state in Eq.~\eqref{eq:fuel_state}, the fundamental upper bound on work extraction is the non-equilibrium free energy, $\F_{\text{noneq}}^{(1:L)}\coloneqq \beta^{-1} \D(\rho^{(1:L)}\|\gamma^{\otimes L})$.
Accordingly, for any bounded, finite-dimensional Hamiltonian, the non-equilibrium free energy rate of a quantum state generated by a stationary quantum process can be defined as 
\begin{equation}
\label{eq:FE_rate}
    f_{\text{noneq}} \coloneqq \lim_{L\to\infty}\frac{1}{L}\mathcal{F}^{(1:L)}_{\text{noneq}}~.
\end{equation}

For comparison, a lower bound is provided by an LO agent that maximizes the immediate reward at each step, yielding a total work $\mathcal{W}_{\text{LO}}^{(1:L)}$ and a rate $w_{\text{LO}}$~\cite{huang2023engines}. Together, these quantities form a hierarchy:
\begin{equation}
\F_{\text{noneq}}^{(1:L)} \geq \F_{\text{TO}}^{(1:L)}\geq\mathcal{W}_{\text{LO}}^{(1:L)}, \hquad f_{\text{noneq}} \geq f_{\text{TO}}\geq w_{\text{LO}}~.
\end{equation}

\noindent\textit{Causal dissipation}\m The difference between the non-equilibrium free energy and the TOFE is generally non-zero, even if entanglement is absent. Here we show that this gap can be quantified by \emph{causal dissipation}, $\delta({Q_{\overrightarrow{1:L}}})$, where
\begin{equation}
\label{eq:causal_definition}
    \delta({Q_{\overrightarrow{1:L}}}) \coloneqq \mathcal{F}_{\text{noneq}}^{(1:L)}-\mathcal{F}_{\text{TO}}^{(1:L)}~.
\end{equation}

Operationally, sequential work extraction is mathematically equivalent to sequential quantum measurement in our framework: rather than directly measuring $Q_t$, the agent implements a thermal operation tailored to target state $\rho_{\Lambda^*(t)}$. The amount of extracted work $W_t$ thus serves as an indirect measurement of the system itself. For a degenerate Hamiltonian, employing the protocol of Ref.~\cite{skrzypczyk2014work} tailored to $\rho=\sum_{i}\lambda_i\ket{\lambda_i}\!\bra{\lambda_i}$ yields an extracted work $W_t$ with the following probability distribution when applied to a state $\sigma$:
\begin{equation}
\label{eq:work_dist}
\Pr(W_t=w^{(i)}) = \bra{\lambda_i}\sigma\ket{\lambda_i},\hquad\beta w^{(i)} \coloneqq \ln{\lambda_i}+\ln d~,
\end{equation}
where $d$ is the dimension of $\mathcal{H}_Q$. The statistics of $W_t$ therefore perfectly mirror the informational statistics of a direct measurement in the $\ket{\lambda_i}$ basis. In what follows, we assume that this measurement is rank-1 for simplicity; that is, each work value precisely identifies the associated state $\ket{\lambda_i}$ onto which the system is projected. In Appendix~\ref{sec:local_perturb}, we show that situations where multiple states $\ket{\lambda_i}$ yield the same work value can be reduced to the rank-1 scenario via small perturbations to the extraction protocol, incurring negligible thermodynamic cost. The direct correspondence between work costs and measurement outcomes then allows us to quantify the thermodynamic cost of causality using information entropy.
\begin{theorem}
\label{thm:thm3}
    For a sequence of multi-time quantum states $\rho^{(1:L)}$, the causal dissipation, $\delta(Q_{\overrightarrow{1:L}})$ in Eq.~\eqref{eq:causal_definition}, can be expressed as
    \begin{equation}
\begin{split}
\frac{\delta(Q_{\overrightarrow{1:L}}) }{k_\mathrm{B}T}&=\min_{\Lambda}\mathbb{E}_{\mathbf{H}_{L-1}|\Lambda} \bigg[\sum_{t=1}^{L-1}H(W_t|\mathbf{H}_{t-1})\\
    &\quad\quad\quad+S(\tilde{\rho}_{Q_L | \mathbf{H}_{L-1}})\bigg]-S(\rho^{(1:L)})\\
\end{split}
\end{equation}
where \(H\) and \(S\) are Shannon and von Neumann entropies. The expectation is over the policy $\Lambda$ mapping each specific belief state $k_{t-1}$ to an action $a_t$. Here, $W_t$ is the extracted work, and $\tilde\rho_{Q_L|\mathbf{h}_{L-1}}$ is the final reduced post-measurement state given an observed history $\mathbf{h}_{L-1}\in\mathbf{H}_{L-1}$.
\end{theorem}


The causal dissipation has the following operational interpretation.
At each time step $t\in[1,L-1]$, the agent measures the local subsystem $Q_t$ using the projective measurement $\Pi_{Q_t|\mathbf{h_{t-1}}}$ based on prior measurements and outcomes. This is followed by work extraction from the collapsed state at $Q_t$. The first Shannon entropy term thus describes the increase in entropy of the agent's memory based on outcome $W_t$. $\tilde{\rho}_{Q_L | \mathbf{H}_{L-1}}$ is the reduced state at time $L$. The von Neumann entropy of the reduced state then represents the work extraction deficit when the agent performs work extraction without further measurements. Causal dissipation is therefore the minimal difference between the total entropy introduced by these measurements and the joint state's original entropy. The equivalence between causal dissipation and this ``operational interpretation" is proved in Appendix~\ref{app:prove_dissipation}. 

For a bipartite system ($L\!=\!2$), causal dissipation reduces to a measure of quantum discord~\cite{zurek2003quantum, brodutch2010quantum}. It vanishes for states with no quantum correlations (e.g., product states or classical-quantum states) and is asymmetric with respect to time ordering, i.e., $\delta(Q_{\overrightarrow{1:L}})\neq\delta(Q_{\overleftarrow{1:L}})$. A unique feature of causal dissipation is that its asymptotic rate can vanish even when the total dissipation over a finite time is non-zero. This, along with other properties, is derived in Appendix~\ref{sec:dissipation_property}. This can also be viewed as a quantum extension of the cost of modularity~\cite{boyd2018thermodynamics}.

\noindent\textit{Perturbed coin}\m 
To illustrate our framework, we benchmark the DP agent against both the unconstrained free energy bound and a greedy local-optimizing (LO) agent when extracting energy from a qubit sequence generated by the ``perturbed coin" source in Fig.~\ref{fig:combined_fig}(a). Each qubit can be in the state $\ket{\phi_0}$ or $\ket{\phi_1}$ depending on a latent state of the HMM analogous to a coin. Specifically, at each time step, the coin flips with probability $p$ and emits state $\ket{\phi_{x_t}}$ depending on whether it lands on $s$ or $s'$. For the work extraction, we implement the protocol proposed in \cite{skrzypczyk2014work}, which yields the work distribution in Eq.~\eqref{eq:work_dist}. The belief states for this model are $\mathcal{K}= \left\{ \left( 1/2 + \epsilon, 1/2 - \epsilon \right) \ \middle|\ \epsilon \in \left[-1/2, 1/2\right] \right\}$, where $\epsilon\!\to\!\pm1/2$ indicates certainty in either of the latent states.

We first compare the TOFE rate ($f_{\text{TO}}$) against the ultimate unconstrained free energy rate ($f_{\text{noneq}}$) (Fig.~\ref{fig:LOvsDPP}(a)). The gap between them precisely quantifies the asymptotic rate of causal dissipation. We compare these for various transition probabilities $p$, and state overlaps $r\! =\! |\!\ip{\phi_0}{\phi_1}\!|^2$ (or fidelity for non-pure states). We see that the dissipation rate vanishes whenever the dynamics are deterministic ($p\! \in\! \{0,1\}$), purely classical ($r\!=\!0$), or uncorrelated ($p\!=\!1/2$ or $r\!=\!1$)---signaling that our causally constrained agent achieves perfect efficiency in regimes where there are no explicitly quantum correlations.
\begin{figure}
    \centering
    \includegraphics[width=1\linewidth]{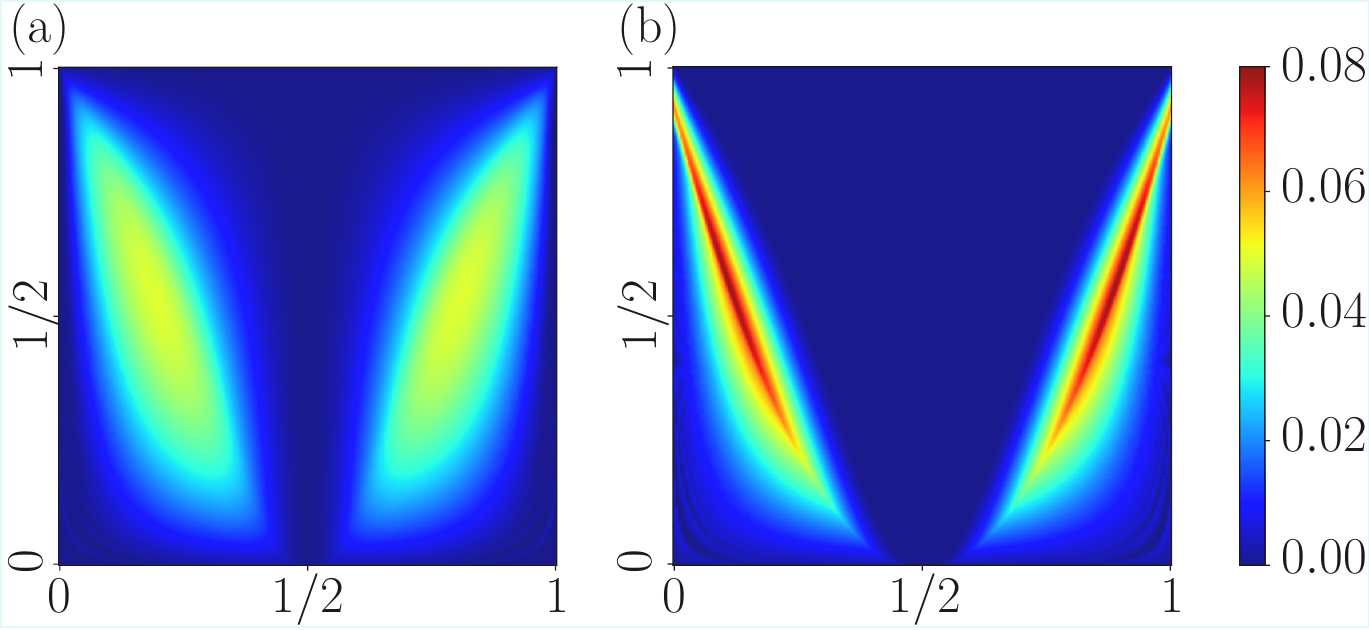}
    \caption{\justifying Comparison of asymptotic work-extraction rates in units of $k_\mathrm{B}T$. Rates are plotted as a function of the transition probability $p$ ($x$-axis) and the overlap between emitted quantum states $r$ ($y$-axis). (a) Work rate lost due to causal constraints, given by $f_{\text{noneq}} - f_{\text{TO}}$. (b) Performance advantage of the optimal policy over a greedy strategy, given by $f_{\text{TO}} - w_{\text{LO}}$.}
    \label{fig:LOvsDPP}
\end{figure}
Conversely, when comparing the optimal policy to the greedy LO strategy (Fig.~\ref{fig:LOvsDPP}(b)), the DP agent's performance advantage ($f_{\text{TO}} - w_{\text{LO}}$) peaks in moderately stochastic regimes. Here, the DP agent strategically incurs immediate dissipation to obtain predictive information, significantly boosting long-term yield. As the process approaches the trivial limits of maximum randomness or indistinguishability, this advantage naturally vanishes, and both strategies converge.

Numerical simulations for finite horizons ($L\!=\!3,4$) confirm that the work deficit exactly matches the predicted causal dissipation (Appendix~\ref{sec:numerical}, Fig.~\ref{fig:full_compare}). Furthermore, while a closed-form expression for $f_{\text{noneq}}$ remains intractable for this model, we rigorously establish its analytical lower bound using the data processing inequality (Appendix~\ref{app:free_energy_rate_bound}).


To unpack the physical mechanics of the optimal policy $\Lambda^*$, in Fig.~\ref{fig:compare_action} we analyze the expected immediate dissipation in Eq.~\eqref{eq:final_rewards} against the agent's uncertainty, quantified by the binary entropy of its belief, $h_2(\gbm\eta_\epsilon)\!=\!-(1/2\!+\!\epsilon)\log_2(1/2\!+\!\epsilon)\!-\!(1/2\!-\!\epsilon)\log_2(1/2\!-\!\epsilon)$. For unpredictable processes ($p \!\approx\! 1/2$, $r\!\approx\!1$), future information has negligible value; the agent defaults to a greedy, low-dissipation strategy across all beliefs. In contrast, in predictable regimes, the agent strategically sacrifices immediate work to gather predictive information. This dissipation peaks at the highest uncertainty ($h_2(\gbm\eta_\epsilon)\!=\!1$) and decreases as beliefs approach certainty ($h_2(\gbm\eta_\epsilon)\! \to\! 0$), explicitly demonstrating how the optimal policy trades immediate thermodynamic yield for long-term predictive power.

\begin{figure}
    \centering
    \includegraphics[width=1\linewidth]{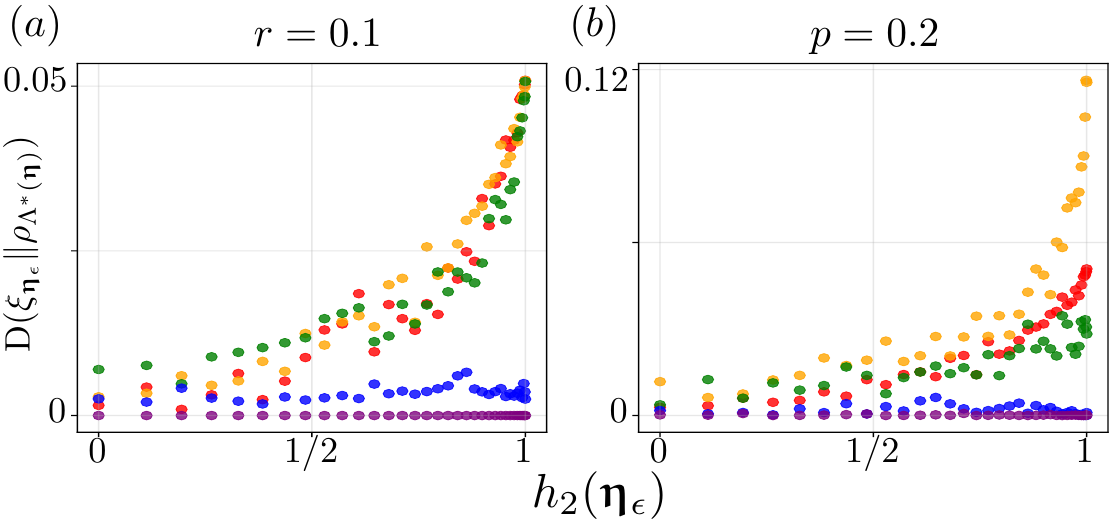}
    \caption{\justifying Relation between how much energy an agent is willing to locally sacrifice vs. uncertainty. The energetic price of information, quantified by $\D(\xi_{\gbm\eta_\epsilon}\|\rho_{\Lambda^*(\gbm\eta)})$ ($y$-axis), represents how much energy an agent is willing to sacrifice when extracting work from the current quantum system. Meanwhile, $h_2(\gbm\eta_\epsilon)$($x$-axis) represents the agent's uncertainty about the underlying source. Panels (a) and (b) show scatter plots of how these quantities relate for various parameters of a perturbed coin source. (a) Varying transition probability $p\! \in\! \{0.1, 0.2, 0.3, 0.4, 0.5\}$ for a fixed overlap $r=0.1$. (b) Varying overlap $r \!\in\! \{0.1, 0.3,0.5, 0.7, 0.9\}$ for a fixed transition probability $p=0.2$. In both panels, colors denote ascending parameter values (red, orange, green, blue, purple). We see that there is a positive correlation between the two quantities: in predictable regimes, the optimal strategy dictates that agents who are less certain of their source should sacrifice more energy to reduce that uncertainty. However, there do also exist unpredictable regimes (blue, purple) where sacrificing energy has little value since it yields little information.}  
    \label{fig:compare_action}
\end{figure}
\noindent\textit{Discussion}\m
Here, we considered the problem of an agent sequentially extracting work from a temporal sequence of $L$ quantum systems generated by a non-Markovian environment. The dynamics of the environment depend on hidden internal degrees of freedom that the agent cannot directly observe. Meanwhile, the agent operates online and is causally constrained: it must choose a local work-extraction protocol based solely on its history of past observations. Using a dynamic programming approach, we demonstrated how to identify the optimal work-extraction policy with computational cost that scales linearly with $L$. This enables us to formally define and determine the time-ordered free energy (TOFE): the true maximum work that can be harvested from a sequence of quantum states subject to causal constraints. Our analysis yields two key insights. First, imposing temporal causality carries a fundamental energetic cost, resulting in an unavoidable work deficit. Second, sequential work extraction demands a fundamental thermodynamic trade-off: to maximize long-term yield, autonomous agents must strategically sacrifice immediate energy to acquire predictive information for future extraction.

These results open several promising avenues for further research. The first is a generalization of our work to systems with non-degenerate Hamiltonians. We made some progress in this direction in Appendix~\ref{sec:nondegen}, showing that the relation between heat dissipation and agent-expectation mismatch continues to hold, while the physical control cost remains unresolved. Meanwhile, a key assumption here is that our agents begin with a full mathematical description of the hidden Markov source that generates the sequence of incident quantum systems, so that learning focuses on inferring which hidden state the source occupies at each time step. It would also be interesting to consider scenarios in which the dynamics of the hidden Markov source itself must be learned, as in recent quantum-state-agnostic work on work extraction~\cite{watanabe2024black,watanabe2026universal,lumbreras2025quantum}. More fundamentally, the heat dissipation imposed by causal constraints provides a new source of irreversibility. This mirrors several studies---both in machine learning and in stochastic modeling---that have found that certain classical sequences are harder to predict in one temporal direction than in the other~\cite{thompson2018causal,kechrimparis2023causal}. It would certainly be interesting to investigate such effects in our energy-harvesting context.

\noindent\textit{Acknowledgments}\m
This work is supported by the National Research Foundation through the NRF Investigatorship on Quantum-Enhanced Agents (Grant No. NRF-NRFI09-0010) and the National Quantum Office, hosted at A*STAR, under its Centre for Quantum Technologies Funding Initiative (S24Q2d0009), the Singapore Ministry of Education Tier 1 Grant RT4/23 and RG91/25 (S) and the RIE25 Japan-Singapore Joint Call on Quantum R25J4IR111.

\bibliographystyle{unsrt}
\bibliography{References.bib}

\setcounter{secnumdepth}{3}
\onecolumngrid
\appendix
\section{Energy Non-degenerate Hamiltonian}
\label{sec:nondegen}
In the main body, we have discussed the application of the $\rho^*$-ideal protocol for a degenerate Hamiltonian. We mentioned that for a degenerate Hamiltonian, this class of operations overlaps with thermal operations, i.e., there exists a thermal operation that can extract all non-equilibrium free energy from any quantum state. This is possible only because every basis is an energy eigenbasis in a degenerate Hamiltonian. The strict energy conservation requirement has famously made thermal operations unable to extract work from coherence with respect to the energy eigenbasis; this vanishes when all bases are eigenbases. Suppose we relax this constraint; then we essentially allow an operation that implements a rotation in the energy eigenbasis, opening up the possibility of creating coherence in the energy eigenbasis, which necessarily violates global energy conservation since it breaks time-translation symmetry~\cite{lostaglio2015quantum}. This was mitigated by the $\rho^*$-ideal protocol whenever the input is diagonalized in $\rho^*$. But in general, the implementation of such an operation on an unknown state will cause coherence to be created in the energy eigenbasis. This would require a further injection of energy. 

Take the example of the protocol we considered, assuming that the Hamiltonian is now $H_Q = E_0 \ket{E_0}\!\bra{E_0}+E_1\ket{E_1}\!\bra{E_1}$. Suppose we apply a protocol tailored to $\rho=\sum_i\lambda_i\ket{\psi_i}\bra{\psi_i}$ to $\sigma$. In the first stage of rotation, the unitary $U = \sum_i \ket{E_i}\!\bra{\psi_i}\otimes \Gamma_{\epsilon_i}$, where $\epsilon_i = \bra{\psi_i}H\ket{\psi_i}-E_i$, is applied. The resultant state $\sigma'_{QB}$ is given by
\begin{equation}
    \begin{split}
        \sigma'_{QB}&=U\sigma_{QB}U^\dag = \sum_{i,j}\ket{E_i}\!\bra{\psi_i}\sigma \ket{\psi_j}\!\bra{E_j}\otimes \Gamma_{\epsilon_i}\rho_W\Gamma_{\epsilon_j}^\dag\\
        &=\sum_{i,j} \bra{\psi_i}\sigma \ket{\psi_j}\ket{E_i}\!\bra{E_j}\otimes \Gamma_{\epsilon_i}\rho_W\Gamma_{\epsilon_j}^\dag
    \end{split}
\end{equation}
Here we can compare the energy before and after the rotation. Before the rotation, the energy of the joint system is given by $\tr(\sigma_{QB}\mathcal{H}_{QB})=\sum_{ij}\bra{\psi_i}\sigma\ket{\psi_j}\bra{\psi_j}\mathcal{H}_Q\ket{\psi_i}$. On the other hand, after the rotation, the energy of the joint system will be given by $\sum_i \bra{\psi_i}\sigma\ket{\psi_i}\bra{\psi_i}\mathcal{H}_Q\ket{\psi_i}$. There is an overall change in energy given by 
\begin{equation}
    \begin{split}
        \Delta U &= \sum_i \bra{\psi_i}\sigma\ket{\psi_i}\bra{\psi_i}\mathcal{H}_Q\ket{\psi_i} - \sum_{i,j}\bra{\psi_i}\sigma\ket{\psi_j}\bra{\psi_j}\mathcal{H}_Q\ket{\psi_i}\\
        &= - \sum_{i\neq j }\bra{\psi_i}\sigma\ket{\psi_j}\bra{\psi_j}\mathcal{H}_Q\ket{\psi_i}~,
    \end{split}
\end{equation}
This also means that the implementation of this unitary is not free but requires \emph{at least} $\Delta U$ of work. Note that if $\sigma$ is also diagonal in the eigenbasis of $\rho$, then $\Delta U = 0$.
Luckily, this is not the end of the world. The expected work extracted by this protocol is given by 
\begin{equation}
    \begin{split}
    \label{eq:exp_work_nondegen}
        \E(W) &= \sum_i\bra{\psi_i}\sigma\ket{\psi_i}\left(\bra{\psi_i}\mathcal{H}_Q\ket{\psi_i}+\beta^{-1}\ln\lambda_i-F\right)\\
        &= \sum_i\bra{\psi_i}\sigma\ket{\psi_i}\bra{\psi_i}\mathcal{H}_Q\ket{\psi_i} + \beta^{-1}\tr(\sigma\ln\rho) -F\\
    \end{split}
\end{equation}
Notice that if we subtract $\Delta U$ from Eq.~\eqref{eq:exp_work_nondegen}, we arrive precisely at
\begin{equation}
\label{eq:recover}
\begin{split}
    \E(W_\text{net})&= \sum_{i,j}\bra{\psi_i}\sigma\ket{\psi_j}\bra{\psi_j}\mathcal{H}_Q\ket{\psi_i} +\beta^{-1}\tr(\sigma\ln\rho) -F\\
    & = \beta^{-1} \left[-\tr(\sigma\ln\gamma)+\tr(\sigma\ln\rho)\right]\\
    &= \beta^{-1}\left[\D(\sigma\|\gamma)-\D(\sigma\|\rho)\right]
\end{split}
\end{equation}
But notice that we previously established that the unitary in stage 1 costs \emph{at least} $\Delta U$. This means that the equality in Eq.~\eqref{eq:recover} is achieved only if the rotation in stage 1 can be implemented at its minimal energetic cost. 

A possible experimental setup was proposed in \cite{skrzypczyk2014work}, which is based on the implementation of a time-dependent interaction Hamiltonian. This, in turn, would require additional resources to keep track of time and change the Hamiltonian smoothly \cite{woods2023autonomous}. More recent work has shown that such an operation can only be realized if one has access to external sources of coherence \cite{aaberg2014catalytic,korzekwa2016extraction}, and even then, this can be achieved approximately using lasers; the energy from the laser must also be accounted for. We leave the careful accounting of these resources for future work and instead focus on the energy-degenerate scenario in the main paper.

\section{Restoring Observability via State Perturbation}

\label{sec:local_perturb} In the context of sequential extraction, the agent relies on the observed work $W_t$ to perform Bayesian updates on its belief state. In certain fine-tuned settings, the observed work value $W_t$ does not uniquely identify the underlying energy eigenstate. Different energy eigenstates may yield the same extracted work value regardless of whether the Hamiltonian is degenerate or non-degenerate. 

Consider a Hamiltonian $\mathcal{H} = \sum_{m} E_m \ket{E_m}\! \bra{E_m}$, where $\{\ket{E_m}\}$ is the energy eigenbasis. If the target state $\rho$ is not initially diagonal in this basis, the agent must first implement a unitary rotation to diagonalize it, paying the associated thermodynamic cost already detailed in Appendix~\ref{sec:nondegen}. Having decoupled that cost, we can henceforth assume the state from which work is extracted is diagonal. Suppose the quantum state is a qudit $\rho=\sum_{i=1}^d \lambda_i\ket{E_i}\!\bra{E_i}$. Following the protocol in \cite{skrzypczyk2014work}, the extracted work value takes the form:

\begin{eqnarray}
    w^{(i)} = E_i+k_BT\ln\lambda_i-\F_{\text{eq}}, \quad \Pr(W=w^{(i)})=\bra{E_i}\rho\ket{E_i}=\lambda_i~.
\end{eqnarray}
Notice that the work values of two distinct energy eigenstates can potentially be identical. Specifically, $w^{(i)} = w^{(j)}$ if and only if:
\begin{equation}
    \lambda_i=\lambda_je^{\beta(E_j-E_i)}~.
\end{equation}
When this condition is met, the agent cannot determine which energy eigenstate it is actually extracting work from. Although this does not affect the average extracted work, it poses a problem if the work value is used as feedback for the agent. To address this, notice that the work value depends on $\lambda_i$. The agent can introduce a small perturbation $\epsilon_i$ to the eigenvalues during extraction. 

Let us denote the perturbed quantum state as $\rho_{\text{perturbed}} = \sum_{i=1}^d (\lambda_i+\epsilon_i)\ket{E_i}\!\bra{E_i}$, where $\sum_{i=1}^d \epsilon_i=0$ to ensure trace preservation. Suppose that the agent tailors the extraction protocol to the perturbed state but operates it on the original state. The work extracted would then be:
\begin{eqnarray}
    w^{(i)} = E_i+k_BT\ln(\lambda_i+\epsilon_i)-\F_{\text{eq}}, \quad \Pr(W=w^{(i)})=\bra{E_i}\rho\ket{E_i}=\lambda_i~.
\end{eqnarray}
\begin{lemma}[Existence of work outcomes via state perturbation]
    For any Hamiltonian $H=\sum_{i=1}^dE_i\ket{E_i}\!\bra{E_i}$ and full-rank target state $\rho=\sum_{i=1}^d\lambda_i\ket{E_i}\!\bra{E_i}$, there exists an arbitrarily small, trace-preserving perturbation $\epsilon=(\epsilon_1,\cdots,\epsilon_d)$ such that the perturbed state $\rho_{\text{perturbed}}=\sum_{i=1}^d(\lambda_i+\epsilon_i)\ket{E_i}\!\bra{E_i}$ yields strictly unique work values for all energy eigenstates.
\end{lemma}
\begin{proof}
    Using the perturbed work values, a degeneracy between any two distinct levels $i\neq j$ occurs if and only if: 
    \begin{eqnarray}
        \epsilon_j-\epsilon_ie^{\beta(E_i-E_j)}=\lambda_ie^{\beta(E_i-E_j)}-\lambda_j~.
    \end{eqnarray}
    Hence, for any fixed pair of levels $(i,j)$, this equation defines a proper $(d-1)$-dimensional affine hyperplane $\mathcal{P}_{i,j}$ in the $d$-dimensional parameter space of all possible perturbations $\epsilon$.

    Since the perturbation must preserve the trace, this constraint restricts the valid perturbations to a $(d-1)$-dimensional subspace. The intersection of any hyperplane $\mathcal{P}_{i,j}$ with this subspace has a dimension of at most $d-2$. Because there are finitely many such hyperplanes, their finite union cannot cover an open set in the subspace. Therefore, there exists a valid, trace-preserving perturbation $\epsilon$ such that $w^{(i)}(\epsilon) \neq w^{(j)}(\epsilon)$ for all $i\neq j$.
\end{proof}
The average energy that one may extract is then, in principle, given by:
\begin{eqnarray}
    \E(W)=k_BT\left[\D(\rho\|\gamma)-\D(\rho\|\rho_\text{perturbed})\right]~,
\end{eqnarray}
where the second term quantifies the excess free energy lost due to the perturbation. In the regime $\vert{}\epsilon_i\vert{} \ll \lambda_i$, we can expand the relative entropy to obtain:
\begin{eqnarray}
    \D(\rho\|\rho_\text{perturbed})\approx\sum_{i=1}^d\lambda_i\left(-\frac{\epsilon_i}{\lambda_i}+\frac{\epsilon_i^2}{2\lambda_i^2}\right) = \sum_{i=1}^d \left(-\epsilon_i +\frac{\epsilon_i^2}{2\lambda_i}\right)~.
\end{eqnarray}
Since we have imposed that $\sum_{i=1}^d\epsilon_i=0$, the excess loss simplifies to
\begin{eqnarray}
    \D(\rho\|\rho_\text{perturbed})\approx \frac{1}{2}\sum_{i=1}^d\frac{\epsilon_i^2}{\lambda_i}
\end{eqnarray}
This demonstrates that the agent can introduce a well-chosen, arbitrarily small perturbation $\epsilon$ to strictly separate all work outcomes and restore perfect Bayesian observability. Meanwhile, the thermodynamic dissipation incurred scales only as $\mathcal{O}(\epsilon^2)$, rendering the cost vanishingly small.

The above analysis assumes that the target state is full-rank. If the state is not full-rank, the assumption $\vert{}\epsilon_i\vert{} \ll \lambda_i$ fails for levels where $\lambda_i=0$. In this scenario, we can first assign a strictly positive perturbation $\epsilon_i> 0$ to any level where $\lambda_i = 0$ to ensure finite work values. The remaining non-zero eigenvalues are then perturbed such that the overall trace is preserved ($\sum_{i=1}^d \epsilon_i = 0$) and all $w^{(i)}$ are unique. The additional loss is given by:
\begin{eqnarray}
   \D(\rho\|\rho_\text{perturbed}) = \sum_{i=1}^d \lambda_i \ln \left( \frac{\lambda_i}{\lambda_i + \epsilon_i} \right) 
\end{eqnarray}
We can split this sum into the support $\mathcal{S}$ and the null space $\mathcal{N}$:
\begin{eqnarray}
\D(\rho\|\rho_\text{perturbed}) = \left[ \sum_{i \in \mathcal{S}} \lambda_i \ln \left( \frac{\lambda_i}{\lambda_i + \epsilon_i} \right) + \sum_{k \in \mathcal{N}} 0 \cdot \ln \left( \frac{0}{\epsilon_k} \right) \right]~.
\end{eqnarray}

The contribution to the relative entropy from the null space spanned by these eigenvectors vanishes ($\lim_{\lambda \to 0} \lambda \ln(\lambda/\epsilon_0) = 0$), by performing a Taylor expansion over the remaining spectrum, we obtain 
\begin{equation}
\begin{split}
    \D(\rho\|\rho_\text{perturbed}) &= \left(\sum_{i\in\mathcal{S}}-\epsilon_i + \sum_{i\in\mathcal{S}}\frac{\epsilon_i^2}{2\lambda_i}\right)\\
    &=\left(\sum_{n\in\mathcal{N}}\epsilon_n+\sum_{i\in\mathcal{S}}\frac{\epsilon_i^2}{2\lambda_i}\right)
\end{split}
\end{equation}
where we used $\sum_{n\in\mathcal{N}}\epsilon_n+\sum_{i\in\mathcal{S}}\epsilon_i=0$. Consequently, for a non-full-rank target state, the perturbation incurs a cost linear in $\epsilon$ for all initially unoccupied levels and quadratic for all occupied levels. 
\section{Proof of Equivalence}
\label{App:equivalence}
We employ dynamic programming (DP) to find the optimal policy, $\Lambda^*$. To do so, we first write the expectation of cumulative work over $L$ time steps, starting from an initial belief $\gbm\eta^{(i)}$ as a value function.
\begin{equation}
\label{eq:overall}
\tilde{V_1}\left(\gbm\eta^{(i)},\Lambda\right) \coloneqq \mathbb{E}_{K_{1:L-1}}\left[\sum_{t=1}^{L} \langle W(K_{t-1},A_t)\rangle \Bigg|K_0=\gbm\eta^{(i)}\right]~,
\end{equation}
where $A_t = \Lambda(K_{t-1})$.  

Here, we provide a formal proof for the equivalence between the TOFE, $\mathcal{F}_{TO}^{(1:L)}$, and the optimal value function $\tilde V^*_1(\gbm\eta^{(i)})$, under the condition that the agent's initial belief is the initial distribution $\mu_0$ of the underlying HMM.
The core of this proof rests on two key ideas:
\begin{enumerate}
    \item The agent's belief state is a sufficient statistic for the entire history of actions and observations.
    \item The belief trajectory $K_{0:L-1}$ is a deterministic function of the trajectory over physical realizations $X_{1:L}$, actions $A_{1:L}$, and observations $W_{1:L}$.
\end{enumerate}

\subsection{The Belief State as a Sufficient Statistic} \label{sec:sufficient_statistic}
We will first prove that belief states are statistically sufficient for making any prediction. Formally, we want to show that
\begin{equation}
    \Pr(X_t|W_{1:t-1},A_{1:t-1},\mu_0)=\Pr(X_t|K_{t-1},K_0=\mu_0)~.
\end{equation}
This tells us that any information that the past observations and actions carry about the future quantum states can be encoded into the belief state, and hence the agent can track the belief state instead of the entire history. 
\begin{lemma}[Sufficiency of Belief States]
\label{lemma:sufficiency}
The belief state $K_{t-1}$ constitutes a sufficient statistic of the history of observations and actions, $(W_{1:t-1},A_{1:t-1})$, for predicting the future output $X_t$. Given an initial distribution over the latent state $\mu_0$, the conditional probability satisfies 
\begin{equation}
    \Pr(X_t|W_{1:t-1},A_{1:t-1},\mu_0)=\Pr(X_t|K_{t-1},K_0=\mu_0)~.
\end{equation}
\end{lemma}
\begin{proof}
    We show this inductively. The base case is trivially true for $t=1$. 
    \begin{equation}
        \Pr(X_1|\mu_0)=\Pr(X_1|K_0=\mu_0)
    \end{equation}
    The LHS of the equation can easily be expressed using the labeled transition matrix $\mathsf{T}^{(x)}$:
    \begin{equation}
        \Pr(X_1=x|\mu_0) = \mu_0 \mathsf{T}^{(x)}\mathbf{1}
    \end{equation}
    whereas the RHS can likewise be expressed in a similar form:
    \begin{equation}
        \Pr(X_1=x|K_0) = K_0\mathsf{T}^{(x)}\mathbf{1}
    \end{equation}
    Now clearly, the two expressions are equivalent as long as we set the condition $K_0 = \mu_0$. The physical intuition of this is simple: before observing any emission, the agent knows how the process is initialized, therefore having a prior in the form of $\mu_0$. If the agent does not have information about $\mu_0$, the physical trajectory can be assumed to begin in the stationary distribution, which represents the long-run average after transient effects have died out~\cite{crutchfield1989inferring,Shalizi2001}. 

    Now, we show this for the base case $t=2$, i.e.,
    \begin{equation}
        \Pr(X_2|W_1,A_1,\mu_0)=\Pr(X_2|K_1,K_0=\mu_0)
    \end{equation}
    From the case of $t=1$, we can rewrite the LHS as
    \begin{equation}
        \Pr(X_2|W_1,A_1,\mu_0)=\Pr(X_2|W_1,A_1,K_0=\mu_0)
    \end{equation}
We now show that given a belief $K_{t-1}$, an action $A$, and an observation $W_t$, one can construct the next belief state $K_t$:
    \begin{equation}
\label{eq:belief_traj}
\begin{split}
   [K_{t}]_s &= \frac{\sum_{x,s'}\Pr(S_t=s,X_t=x,W_t=w|S_{t-1}=s',A_t=a)[K_{t-1}]_{s'}}{\sum_{x,s,s'}\Pr(S_t=s,X_t=x,W_t=w|S_{t-1}=s',A_t=a)[K_{t-1}]_{s'}}\\
   &=\frac{\sum_{x,s'}\mathsf{T}^{(x)}_{s',s}\Pr(w|x,a)[K_{t-1}]_{s'}}{\sum_{x,s,s'}\mathsf{T}^{(x)}_{s',s}\Pr(w|x,a)[K_{t-1}]_{s'}}\\
   &=\Pr(S_t=s|W_t=w,A_t=a,K_{t-1})~,
\end{split}
\end{equation}
where the summations are over $x \in \mathcal{X}$ and $s,s'\in\mathcal{S}$.
The sum in the denominator is the probability $\Pr(W_t=w|K_{t-1},A_t=a)$.
Notice that the dependence on $X_t$ does not come into play explicitly during the update, so given any sequence of $A_{1:t}$ and $W_{1:t}$, it is always possible to construct a belief trajectory $K_{0:t}$ that faithfully captures the posterior probability distribution over the latent states.
    
    Using this, we can explicitly construct $K_1$ given $K_{0}$, $W_1$ and $A_1$. Hence we can obtain that
    \begin{equation}
        \Pr(X_2|W_1,A_1,\mu_0)=\Pr(X_2|W_1,A_1,K_0=\mu_0) = \Pr(X_2|K_1,K_0=\mu_0)
    \end{equation}
    Next, assuming that this equivalence holds for an arbitrary time step $t$, 
    \begin{equation}
    \label{eq:assume}
    \Pr(X_t|W_{1:t-1},A_{1:t-1},\mu_0)=\Pr(X_t|K_{t-1},K_0=\mu_0)~.
    \end{equation}
    We show that this also holds for $t+1$:
    \begin{equation}
        \Pr(X_{t+1}|W_{1:t},A_{1:t},\mu_0) = \Pr(X_{t+1}|K_t,K_0=\mu_0)
    \end{equation}
    We first simplify the RHS using Eq.~\eqref{eq:assume}
    \begin{equation}
        \begin{split}
        \Pr(X_{t+1}|W_{1:t},A_{1:t},\mu_0) &= \Pr(X_{t+1}|W_t,A_t,K_{t-1},K_0=\mu_0)
        \end{split}
    \end{equation}

We can expand the term $\Pr(X_{t+1}|W_t,A_t,K_{t-1},K_0=\mu_0)$ by marginalizing over the latent state $S_t$:
\begin{equation}
    \Pr(X_{t+1}|W_t,A_t,K_{t-1},K_0=\mu_0)=\sum_{s\in \Ss}\Pr(X_{t+1}|S_t=s)\Pr(S_t=s|W_t,A_t,K_{t-1},K_0=\mu_0)
\end{equation}
From Eq.~\eqref{eq:belief_traj}, we know that $\Pr(S_t=s|W_t,A_t,K_{t-1},K_0=\mu_0)$ is precisely the definition of the $s$-th element of the updated belief state $[K_t]_s$. Furthermore, due to Markovianity, $X_{t+1}$ depends only on the history through the current latent state $S_t$. Therefore  
\begin{equation}
        \begin{split}
    \Pr(X_{t+1}|W_{1:t},A_{1:t},\mu_0) &= \Pr(X_{t+1}|W_t,A_t,K_{t-1},K_0=\mu_0)\\
        &=\sum_{s\in \Ss}\Pr(X_{t+1}|S_t=s)[K_t]_s\\
        &=\Pr(X_{t+1}|K_t,K_0=\mu_0)~.
        \end{split}
    \end{equation}
\end{proof}
\subsection{Formal Equivalence}
\label{sec:equivalence}
While the technique of replacing observation histories with belief states is standard in POMDP literature~\cite{smallwood1973optimal,puterman2014markov}, we provide an explicit derivation below adapted to our thermodynamic framework, directly linking the global multi-time state to the agent's value function.

The quantity we wish to maximize is the average cumulative work extracted over $L$ time steps, defined as the time-ordered free energy (TOFE):
\begin{align}
    \F^{(1:L)}_{\text{TO}}(\rho^{(1:L)})\!\coloneqq\!\max_{\Gamma}\E_{\mathbf{H}_L|\Gamma}\left[\sum_{t=1}^LW_t\right]~.
\end{align}
By the linearity of expectation, we can express the total expected work as the sum of the expected work at each time step:
\begin{align}
    \F^{(1:L)}_{\text{TO}}(\rho^{(1:L)})\!\coloneqq\!\max_{\Gamma}\sum_{t=1}^L\E_{\mathbf{H}_L|\Gamma}\left[W_t\right]~.
\end{align}
The immediate work $W_t$ at time $t$ depends strictly on the underlying post-processed quantum state of subsystem $Q_t$ and the agent's action $A_t = \Gamma(\mathbf{H}_{t-1})$, which is dictated by the observable history up to that point. Using the law of total expectation, we can marginalize over the future and condition the immediate work exclusively on the history $\mathbf{H}_{t-1}$:
\begin{align}
    \F^{(1:L)}_{\text{TO}}(\rho^{(1:L)})\!\coloneqq\!\max_{\Gamma}\sum_{t=1}^L\E_{\mathbf{H}_{t-1}|\Gamma}\left[\E_{\mathbf{H}_{L}|\Gamma}(W_t|\mathbf{H}_{t-1},A_t= \Gamma(\mathbf{H}_{t-1}))\right]~.
\end{align}
As established in Lemma~\ref{lemma:sufficiency}, the belief state $K_{t-1}$ is a sufficient statistic for the observable history $\mathbf{H}_{t-1}$ regarding the future output of the system. Furthermore, as shown in Eq.~\eqref{eq:belief_traj}, the belief state can be iteratively updated given an initial prior $K_0=\mu_0$ and the specific history trajectory. Therefore, any probability distribution over the histories $\mathbf{H}_{t-1}$ uniquely induces a probability distribution over the belief states $K_{t-1}$. 

From Equation~\eqref{eq:final_rewards}, we see that the expected work $W_t$ depends only on the action $A_t$ and the distribution of $X_t$. Because $K_{t-1}$ is a sufficient statistic, it predicts the same distribution of $X_t$ as the full history $\mathbf{H}_{t-1}$. Notice, however, that the space of possible histories is generally much larger than the space of belief states; multiple distinct histories can map to the same belief state $K_{t-1}$. Since expected work depends on history \emph{only} through the distribution of $X_t$, a policy gains no advantage by distinguishing between histories that yield the same belief. Consequently, any history-dependent policy can be replaced---without loss of expected work---by a deterministic, belief-dependent policy that assigns the same action to all realizations of $\mathbf{H}_{t-1}$ sharing the same $K_{t-1}$~\cite{aastrom1965optimal, smallwood1973optimal}. This sufficiency allows us to perfectly replace conditioning on the full, growing history with conditioning on the compact belief state.


\begin{equation}
    \E_{W_{t}|\Gamma}
    [W_t|\mathbf{H}_{t-1},A_t]
    =\E_{W_{t}|\Lambda}
    [W_t|K_{t-1},A_t=\Lambda(K_{t-1}(\mathbf{H}_{t-1}))]~,
\end{equation}
where $\Lambda$ is the belief-dependent policy that maps beliefs to actions.

By definition in Eq.~\eqref{eq:final_rewards}, the expected immediate work extracted when taking action $A_t$ while possessing belief state $K_{t-1}$ is exactly $\langle W(K_{t-1},A_t)\rangle$. Substituting this into our expectation over the belief states yields:
\begin{equation}
\F^{(1:L)}_{\text{TO}} (\rho^{(1:L)}) = \max_{\Lambda}\sum_{t=1}^L \E_{K_{t-1}|\Lambda}\left[\langle W(K_{t-1},A_t)\rangle|K_0=\mu_0\right]~.
\end{equation}
The condition on $K_0=\mu_0$ reflects the agent's prior, i.e., the inference the agent has with $\mathbf{H}_0\coloneq\emptyset$ given knowledge of the HMM; this is standard in the computational mechanics literature~\cite{ellison2009prediction,travers2011exact,riechers2018spectral}.
Finally, by bringing the summation back inside the expectation, we arrive at the exact definition of the optimal global value function optimized by our dynamic programming algorithm:
\begin{equation}
\F^{(1:L)}_{\text{TO}} (\rho^{(1:L)}) = \max_{\Lambda}\E_{K_{1:L-1}|\Lambda}\left[\sum_{t=1}^L \langle W(K_{t-1},A_t)\rangle\bigg|K_0=\mu_0\right] = \tilde{V_1}^*(\mu_0)~.
\end{equation}

\section{Optimality of Backward Iteration}
\label{App:optimality_DP}
Here we present the explicit algorithm used in backward iteration, followed by the proof of its optimality. The pseudocode of the algorithm can be found in Alg.~\ref{alg:DPP_for_work}. For each time step, the optimization search is performed over the same number of belief states and possible actions. Therefore, the algorithm is $\mathcal{O}(L)$ in the total number of time steps $L$. To implement the dynamic programming approach computationally, we discretize the continuous belief space into $N$ finite states, denoted $\gbm\eta^{(i)}$, and limit the agent to $M$ finite actions, $a^{(j)}$.
\begin{algorithm}[htbp!]
	\caption{\textsf{Pseudocode for work extraction}} 
    \label{alg:DPP_for_work}
  \textbf{input:} discretized set of belief states $\mathcal{K}_{\text{DP}}\coloneqq \{\gbm \eta^{(i)}\}_{i=1}^{N}$, set of possible actions $\A \coloneqq \{a^{(j)}\}_{j=1}^{M}$, reward function $\mathcal{R}(\gbm\eta_{t-1},a_t)$, the underlying HMM, and the total time $L$ of the process.
    \\

    \emph{Set final value to 0 as there are no more future rewards} \\
    Set $V^{*}_{L+1}(\gbm \eta^{(i)}_{L}) \coloneqq 0$ for all $\gbm \eta^{(i)}$ \\

    \For {$t = L, L-1, \cdots, 1$}{
    \For {belief state $\gbm \eta^{(i)}_{t-1} \in \mathcal{K}_{\text{DP}}$ from $i=1$ to $N$}{
    \For {action $a_{t}^{(j)} \in \A$ from $j=1$ to $M$}{
    
    \emph{Find next state following Eq. \eqref{eq:bayesian_update}} \\
    Find $\gbm \eta^{(k)}_{t}$ from $\gbm \eta^{(i)}_{t-1}$, $a_{t}^{(j)}$ and possible values of $w_t$\\

    \emph{Calculate expected value for each action} \\
    $V_{t}(\gbm \eta^{(i)}_{t-1}, a^{(j)}_{t}) \coloneqq \mathcal{R}(\gbm \eta^{(i)}_{t-1},a^{(j)}_t) + \mathbb{E}_{K_t|\gbm\eta^{(i)}_{t-1},a^{(j)}_t}~\left[V^{*}_{t+1}(K_{t})\right]$
    }   

    \emph{Find best value and action from all actions} \\
    Set $V^{*}_{t}(\gbm \eta^{(i)}_{t-1}) \coloneqq \max_{a^{(j)}_{t}}{V_{t}(\gbm \eta^{(i)}_{t-1}, a^{(j)}_{t})}$ \\
    Set $a^{*}_{t}(\gbm \eta^{(i)}_{t-1}) \coloneqq \argmax_{a^{(j)}_{t}}{V_{t}(\gbm \eta^{(i)}_{t-1}, a^{(j)}_{t})}$
    }
    }

    \textbf{output:} policy $\Lambda^* = \{ \mathcal{K}_{\text{DP}_{t-1}}\to \A_{t} \}_{t=1}^{L}$
\end{algorithm}

To prove that the policy obtained via DP is indeed optimal, we first note that the optimal policy maximizes the global value function. As in Eq.~\eqref{eq:overall}, we denote the optimal policy by $\Lambda^* \coloneqq \{a^*_1(\gbm \eta^{(i)}),a^*_2(\gbm \eta^{(i)}), \ldots\}_i$, where each $a^*_t(\gbm \eta^{(i)})$ is the optimal action for belief state $\gbm \eta^{(i)}$ at time $t$. We have previously denoted the highest global value function as $\tilde V_1^*(\gbm\eta^{(i)})$ as the maximum of $\tilde V_1(\gbm\eta^{(i)})$ over all possible actions. Since $\Lambda^*$ is a possible policy, it has its own global value function $\tilde V_1^{(\Lambda^*)}(\gbm\eta^{(i)})$, and hence $\tilde V_1^{(\Lambda^*)} (\gbm\eta^{(i)})\leq \tilde V_1^*(\gbm\eta^{(i)})$. It remains to prove that $\tilde V_1^{(\Lambda^*)}(\gbm\eta^{(i)}) \geq \tilde V_1^*(\gbm\eta^{(i)})$. If this is satisfied, then $\tilde V_1^{(\Lambda^*)}(\gbm\eta^{(i)}) = \tilde V_1^*(\gbm\eta^{(i)})$ and $\Lambda^*$ is the optimal policy that yields the highest global value function. 

For any $a_{1:L}$ in any policy $\Lambda$ and all $\gbm \eta^{(i)}$, the global value function is formally written as
\begin{equation}
    \tilde V_1^{(\Lambda)}(\gbm\eta^{(i)}) = \E\left(\sum_{t=1}^{
    L}\mathcal{R}( K_{t-1},a_t)+V_{L+1}\Bigg|K_0=\gbm\eta^{(i)}\right)
\end{equation}\\
 We can rewrite this as a conditional expectation conditioned on the penultimate belief state, $K_{L-1}$. This is possible because the metadynamics of belief states are Markovian. This can then be done iteratively backward in time, as demonstrated in Eq.~\eqref{eq:proof}.
\begin{equation}
    \begin{split}
        \label{eq:proof}
        V_1^{(\Lambda)}(\gbm\eta^{(i)}) &= \E\left(\sum_{t=1}^{L} \mathcal{R}( K_{t-1},a_t)+V_{L+1}\Bigg|K_0=\gbm\eta^{(i)}\right)\\
        &=\E\left[\E\left(\sum_{t=1}^{L-1}\mathcal{R}(K_{t-1},a_t)+\mathcal{R}(K_{L-1},a_{L})+V_{L+1}\Bigg|K_0=\gbm\eta^{(i)}\right) \Bigg |K_{L-1}\right]\\
        &= \underbrace{\E\left(\sum_{t=1}^{L-1}\mathcal{R}(K_{t-1},a_t)\Bigg |K_0 = \gbm\eta^{(i)}\right)}_{I_{L-1}} +\E\left[\underbrace{\E\left(\mathcal{R}(K_{L-1},a_{L})+V_{L+1}\Bigg| K_{L-1}\right)}_{J_{L}(K_{L-1},a_{L})}\Bigg | K_0=\gbm\eta^{(i)}\right]\\
        &\leq I_{L-1} + \E\left[\max_{a\in\A} J_{L}(K_{L-1},a_{L})\Bigg | K_0=\gbm\eta^{(i)}\right]\\
        &= I_{L-1} + \E\left[J_{L}(K_{L-1},a^*_{L})\Bigg | K_0=\gbm\eta^{(i)}\right]\\
        &= I_{L-1}+\E\left(V_{L}(K_{L-1})\Big|K_0=\gbm\eta^{(i)}\right)\\
        &=\E\left(\sum_{t=1}^{L-2}\mathcal{R}(K_{t-1},a_t)\Bigg |K_0 = \gbm\eta^{(i)}\right)+ \E\Big[\E\left(\mathcal{R}(K_{L-2},a_{L-1})+V_{L}(K_{L-1})\Big|K_{L-2}\right)\Big|K_0=\gbm\eta^{(i)}\Big]\\
        &\leq I_{L-2}+\E\Big[\max_{a\in\A}J_{L-1}(K_{L-2},a_{L-1})\Big|K_0=\gbm\eta^{(i)}\Big]\\
        &=I_{L-2}+\E\left(V_{L-1}(K_{L-2})\Big|K_0=\gbm\eta^{(i)}\right)\\
        &\vdots\\
        &=V_1^{(\Lambda^*)}(\gbm\eta^{(i)})~,
    \end{split}
\end{equation}
where $J_{t} (\gbm\eta^{(i)}, a_{t}) := \mathbb{E} \left[ \mathcal{R} \left(\gbm\eta^{(i)},a_{t}\right) + V_{t+1}(K_{t}) |K_{t-1}=\gbm\eta^{(i)}\right]$ is the objective function and $V_t(\gbm\eta^{(i)}) = \max_{a_{t} \in \mathcal{A}} J_{t} (\gbm \eta^{(i)}, a_{t})$ is the value function of the belief state $\gbm\eta^{(i)}$ at time $t$.
The first inequality comes from the maximization of $J_{L}$ over all actions, and the sixth line follows from the definition of $V_t$. Notice that after each maximization, the optimal action can be generated, forming the optimal policy $\Lambda^*$. This proof extends to $t=1$, thereby proving that $\tilde V_1^{(\Lambda)}(\gbm\eta^{(i)})\leq \tilde V^{(\Lambda^*)}_1(\gbm\eta^{(i)})$ for any possible policy. Since $\tilde V_1^*$ is defined as $\max_{a\in\A} \tilde V^{(\Lambda)}_1(\gbm\eta^{(i)})$, this implies that $\tilde V_1^*(\gbm\eta^{(i)})\leq \tilde V_1^{(\Lambda^*)}(\gbm\eta^{(i)})$. Together with the previous statement that $\tilde V_1^{(\Lambda^*)} (\gbm\eta^{(i)})\leq \tilde V_1^*(\gbm\eta^{(i)})$, this proves that $\tilde V_1^{(\Lambda^*)} (\gbm\eta^{(i)})= \tilde V_1^*(\gbm\eta^{(i)})$. Hence, $\Lambda^*$ is indeed the optimal policy.

\section{Dynamic Programming Principle}
\label{App:DPP}

Now, we introduce the generic Markovian dynamic programming method. From this point onward, all time indices will be written as subscripts, while the exact indices of the states will be written as superscripts in parentheses to avoid confusion.

Dynamic programming is a mathematical optimization method that recursively breaks down a larger problem into simpler subproblems. This systematic decomposition often reduces computational cost at the expense of requiring more memory for the algorithmic implementation. For general stochastic optimal control problems \cite{bellman1965dynamic}, dynamic programming methods have been shown to produce the optimal results. \\

A stochastic optimal control problem is an optimization problem to find a control, which can be viewed as an agent's action, that optimizes an objective function (usually the expectation of the reward or cost function) for a stochastic process over a period of time. To formalize the optimization, an overall process can be defined as an extension of a stochastic process $..., X_{t-1}, X_{t}, X_{t+1}, ...$ with the addition of an agent's action (the controls) at each time step $A_{t}$ that affects the future outcome of the process, resulting in the process that can be written as $..., X_{t-1}, A_{t-1}, X_{t}, A_{t}, X_{t+1}, ...$. \\

For a Markovian process with a discrete finite time horizon from $t=0$ to $t=L \in \mathbb{Z}^{+}$, where the states $x_{t}$ and actions $a_{t}$ at each time $t$ belong to the discrete sets $\mathcal{X}$ and $\mathcal{A}$, the random variable $X_{t}$ depends solely on the previous state $x_{t-1}$ and the previous action $a_{t-1}$. The action $a_{t}$ at any time $t$ is also a simplified form, with the full form being $a_{t}(x_{t})$, as it depends only on the previous state $x_{t}$. Likewise, the reward at each time $\mathcal{R}(x_{t+1}, a_{t}, x_{t})$ depends on the current and next time steps. \\

With these, we can define the objective function to be optimized over as 
\begin{equation} \label{eq:obj-func}
        J(x_{0}) = \mathbb{E} \left[ \sum^{L-1}_{t=0} c^{t} \cdot \mathcal{R}( X_{t}, A_{t}, X_{t+1} ) \right]~,
\end{equation}
where $x_{0}$ is the initial value for $t=0$, $X_{t}$ and $A_{t}$ are random variables for all $t \neq  0$, and $0<c<1$ is a fixed scaling factor, usually used over an infinite time horizon to prevent the objective function from diverging. Since we are working over a finite time horizon, we can usually set $c = 1$. \\

The value function can then be defined as the maximized objective function:
\begin{equation} \label{eq:mark-val-func}
    V(x_{0}) := \sup_{(A_{t})^{L-1}_{t=0} \in \bar{\mathcal{A}}} \mathbb{E} \left[ \sum^{L-1}_{t=0} c^{t} \cdot \mathcal{R}( X_{t}, A_{t}, X_{t+1} ) \right]~,
\end{equation}
where $\bar{\mathcal{A}}$ is the whole action space from $t=0$ to $t=L-1$. This can be quite intractable in general, with $2L$ random variables and $L$ choices to optimize over. Hence, it is common to use dynamic programming to simplify the problem. \\

\noindent \textbf{Iterative approach}\m
In the iterative approach, the value function can be calculated via backward induction to compute the optimal policy $\Lambda^*$. We first define an objective function $J_{t} (x_{t}, a_{t})$ and a value function $V_{t} (x_{t})$ for each time step $t$. The objective function at time $t$ is essentially the expected reward for the future given the current state $x_{t}$ and action $a_{t}$ of the current time step, with $c$ determining the relative weighting of near-term and long-term rewards. Hence, it can be written as
\begin{equation} \label{eq:mark-obj-func-t}
    J_{t} (x_{t}^{(i)}, a_{t}^{(j)}) := \mathbb{E} \left[ \mathcal{R}\left(x_{t}^{(i)},a_{t}^{(j)},X_{t+1}\right) + c V_{t+1}(X_{t+1}) \right]~.
\end{equation} 

This results in the value function $V_{t} (x_{t}) = \max_{a_{t}^{(j)} \in \mathcal{A}} J_{t}(x_{t}^{(i)}, a_{t}^{(j)})$ based on the optimal action of $\hat{a_{t}}(x_{t}) := \argmax_{a_{t}^{(j)} \in \mathcal{A}} J_{t}(x_{t}^{(i)}, a_{t}^{(j)})$. Since all possible rewards would have already been collected at the final time $L$ and no more rewards can be obtained, the final value function $V_{L}(x_{L})$ is usually set to 0 for all $x_{L}$. For the remaining value functions, 

\begin{equation}
    \label{eq:back_induction}
    V_{t}(x_t^{(i)}) = \max_{a_{t}^{(j)}\in\A}\mathbb{E}\left[\mathcal{R}\left(x_{t}^{(i)}, a_{t}^{(j)}, X_{t+1} \right)+ c V_{t+1}\left(X_{t+1} \right)\right] ~.
\end{equation}

We can then recursively calculate the objective function $J_{t} (x_{t}^{(i)}, a_{t}^{(j)})$ from time $L-1$ to 0, maximizing this function to find the value function $V_{t} (x_{t}^{(i)})$ and the optimal action $\hat{a_{t}}(x_{t}^{(i)})$ at time $t$ for the specific state-action pair $x_{t}^{(i)}, a_{t}^{(j)}$. The resulting set of actions defines the optimal agent policy, with $V_{0}(x_{0}^{(i)}) = V(x_{0}^{(i)})$. \\

\section{Analytics for the DP Policy}
\label{App:analytics}

For a stationary process where transition probabilities among states and rewards do not change over time, the optimal policies can still be time-dependent over a finite horizon (finite $L$), due to the presence of boundary effects. However, over an infinite horizon ($L\to\infty$), the optimal policies are time-independent~\cite{sutton1998reinforcement,scherrer2012use}. 
The optimal policy for a finite-horizon MDP can also be divided into two different phases, the stationary and non-stationary phases, which correspond to the time-independent and time-dependent phases, respectively. The non-stationary phase occurs when future rewards diminish towards the end of the process, i.e., when $t$ is close to $L$. This can be illustrated by considering a \emph{two-step process}, i.e., a process that lasts for two time steps.

\subsection{Non-stationary phase}
\label{sec:transient}

At time $t=2$, which is the last step of the process, the value function $V_{3}$ is 0 regardless of the belief state. Hence, the action taken at $t=2$ for a DP agent will minimize the expected dissipation,
\begin{equation} \label{eq:last_step}
 a_{\text{2}}^*(K_1=\gbm\eta^{(i)})=\argmax_{a^{(j)}\in\A}\left[-\frac{1}{\beta}\D(\xi_{\gbm\eta^{(i)}}\|\rho^{(j)})\right]~.
\end{equation}
This achieves the maximum value of 0 when $\rho^{(j)}=\xi_{\gbm\eta^{(i)}}$. Now, we move on to the first time step, $t=1$. The value function that the agent optimizes becomes
\begin{equation}
\begin{split} \label{eq:V1}
    V_1(K_0=\gbm\eta^{(i)}) &= \max_{a^{(j)}\in\A}\bigg\{\frac{1}{\beta}\left[\D(\xi_{\gbm\eta^{(i)}}\|\gamma)-\D(\xi_{\gbm\eta^{(i)}}\|\rho^{(j)})\right]\\
    &+\E\left[\frac{1}{\beta}\D(\xi_{K_1}\|\gamma)\bigg|K_0 = \gbm\eta^{(i)},a^{(j)}\right]\bigg\}
\end{split}
\end{equation}
where the expectation is taken over the conditional distribution $\Pr(K_1|K_0=\gbm\eta^{(i)},a^{(j)})$. Therefore, the action taken by the DP agent at the first time step is given by
\begin{equation}
\begin{split} \label{eq:action_of_DP}
    a_{1}^*(K_0=\gbm\eta^{(i)})&= \argmax_{a^{(j)}\in\A}\bigg\{-\frac{1}{\beta}\D(\xi_{\gbm\eta^{(i)}}\|\rho^{(j)})\\
    &+\E\left[\frac{1}{\beta}\D(\xi_{K_1}\|\gamma)\bigg|K_0 = \gbm\eta^{(i)},a^{(j)}\right]\bigg\}~.
\end{split}
\end{equation}
The first term represents the expected dissipation of the current action $a_1$, while the second term represents the expectation of future reward. This creates a trade-off between the current reward and the future rewards that the DP agent has to optimize. In general, an action at the end of the process differs from those at the start, since the change in expected future reward becomes significant towards the end. This also depends on the expectation of future rewards. The closer the future reward is to a constant function, the smaller the boundary effect.

For the case of the perturbed coin, this can be illustrated in Fig.~\ref{fig:transient_action}.
\begin{figure}
    \centering
    \includegraphics[width = 0.4 \linewidth]{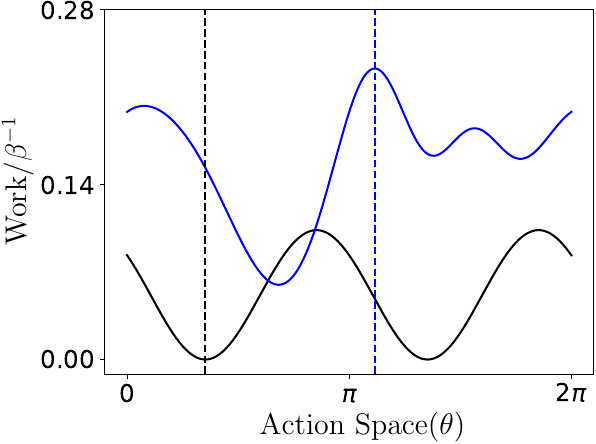}
    \caption{Graph of reward and dissipation, conditioned on the belief state $K = \mu_0=(1/2,1/2)$. The action space is parametrized by $\theta\in[0,2\pi]$. The blue line represents $V_1(K_0=\mu_0)$ as defined in Eq.~\eqref{eq:V1}, and the black line represents the dissipation incurred at the second time step, in Eq.~\eqref{eq:last_step}. The blue and black dotted lines correspond to the optimal actions taken at $t=1$ and $t=2$, respectively.}
    \label{fig:transient_action}
\end{figure}
We can see that the action taken in the second time step $a_2$ is chosen such that the dissipation is minimized, whereas $a_1$ is chosen at the point where the value function $V_1$ is maximized. It is evident that for a finite-horizon optimization, the optimal policy will be time-dependent since $a_2^*(\gbm\eta^{(i)})\neq a_1^*(\gbm\eta^{(i)})$. The closer the expectation of future reward is to a constant function with respect to the actions, the smaller the boundary effect of the non-stationarity of the optimal action.

In the case of the perturbed coin, the expectation reaches a constant in three cases.
\begin{enumerate}
    \item When $p=0.5$, where the process becomes a purely random process. The expected distribution of the future is uniform regardless of the current belief.
    \item When $r=0$, the classical limit where all states are orthogonal; one can measure along the eigenbasis to obtain perfect knowledge.
    \item When $r=1$, a trivial limit where the process emits identical states; hence, regardless of action, the expectation of the future is the same.
\end{enumerate}
Interestingly, when the boundary effect completely disappears even for the finite-horizon case, the optimal policy for all time steps becomes the same as the local-optimizing policy; there is no longer a need to be concerned about the future reward.

\section{Narrowing Down the Search Space}
\label{sec:narrow_search}
In this section, we derive the optimal target states for each protocol. The update of the belief state is entirely dependent on the set of probabilities $\{\bra{\lambda_i}\sigma^{(x)}\ket{\lambda_i}\}_i$. We will first show that for any chosen eigenbasis $\{\ket{\lambda_i}\}_i$ to which the protocol is tailored, the optimal state that results in the least dissipation is given by 
\begin{equation}
\label{eq:best_state}
    \rho_k^*=\sum_ip^*_i\ket{\lambda_i}\bra{\lambda_i},\quad p_i^*=\bra{\lambda_i}\xi_k\ket{\lambda_i} \quad\text{for}\quad i\in \{1,\cdots,d\}
\end{equation}
where $\xi_k$ is the expected state formed from $\gbm\eta^{(k)}$. Then we show that the complex degree of freedom can be ignored when choosing the eigenbases.
\begin{theorem}
\label{thm2}
For a given expected state $\xi_{\gbm\eta^{(k)}}$ and a chosen eigenbasis $\{\psi_i\}_i$, the state $\rho^*$ that is diagonal in this basis and minimizes the expected dissipation, $\D(\xi_{\gbm\eta^{(k)}}\|\rho^*)$, is given by 
\begin{equation}
\label{eq:opt_state}
\rho^* = \sum_i\lambda_i\ket{\psi_i}\!\bra{\psi_i}, \hquad \text{where} \hquad \lambda_i=\bra{\psi_i}\xi_{\gbm\eta^{(k)}}\ket{\psi_i}~.
\end{equation}    
\end{theorem}
\begin{proof}
    First, notice that the Bayesian update is characterized solely by the work distribution:
    \begin{eqnarray}
        \bra{\psi_i}\sigma^{(x)}\ket{\psi_i},\quad \forall x\in\X~.
    \end{eqnarray}
   This also means that if the work distributions for two different protocols, $\mathcal{W}_{\rho}$ and $\mathcal{W}_{\rho'}$, are the same, then all the future statistics should remain the same. The only difference is then the dissipation incurred. 
For any $\rho$ to which the protocol is tailored, we can write the expected dissipation term as
\begin{equation}
    \D(\xi_k\|\rho) = -\tr(\xi_k\ln\rho) - S(\xi_k)~.
\end{equation}
Since the second term is independent of $\rho$, we simply wish to minimize the first term, i.e., we wish to find $\rho^*$ that maximizes $\tr(\xi_k\ln\rho)$. We expand the expression using the fact that $\rho=\sum_i\lambda_i\ket{\lambda_i}\bra{\lambda_i}$:
\begin{equation}
\label{eq:optimality_proof}
    \tr(\xi_k\ln\rho)  = \sum_i\ln \lambda_i\bra{\lambda_i}\xi_k\ket{\lambda_i}
\end{equation}
Now, to ensure that two protocols, $\mathcal{W}_{\rho}$ and $\mathcal{W}_{\rho'}$, induce the same future statistics in the Bayesian sense, we require that $\bra{\lambda_i}\xi_k\ket{\lambda_i} = \bra{\lambda'_i}\xi_k\ket{\lambda'_i}$ for all $i\in\{0,1\}$ and $\{\lambda_i\}_i,\{\lambda'_i\}_i$ are the eigenvectors of $\rho$ and $\rho'$ respectively. Thus, the maximum is achieved by changing the eigenvalues. A simple analysis of Eq.~\eqref{eq:optimality_proof} shows that it reaches maximum value when $\lambda_i = \bra{\lambda_i}\xi_k\ket{\lambda_i}$, which establishes Theorem~\ref{thm2}.
\end{proof}

\section{Bounds on the Free-Energy Rate}
\label{app:free_energy_rate_bound}
It is necessary to show that there is a separation between the best possible adaptive local strategy and the global collective strategy. Therefore, we compare the asymptotic work-extraction rate based on the policy obtained via DP with the free energy rate of the quantum state. In the previous section, we have already discussed the asymptotic rate of the DP-agent; here we show that even the lower bound of the free-energy rate is higher, hence establishing a separation. 

We first define the so-called free energy rate. Consider a quantum state in Eq.~\eqref{eq:fuel_state}.
The free energy rate can then be defined as
\begin{equation}
    r \coloneqq \lim_{L\to\infty}\frac{\beta^{-1}\D(\rho^{(1:L)}\|\gamma^{\otimes L})}{L}~,
\end{equation}
if the limit on the right-hand side exists and is well-behaved. Assuming that this limit exists, we can expand the expressions to obtain 
\begin{equation}
\begin{split}
r=&\lim_{L\to\infty}\frac{\beta^{-1}\D(\rho^{(1:L)}\|\gamma^{\otimes L})}{L}\\
    =& \lim_{L\to\infty}\frac{\beta^{-1}}{L}\left[-\tr(\rho^{(1:L)}\ln\gamma^{\otimes L})-S(\rho^{(1:L)})\right]\\
    =& \lim_{L\to\infty}\frac{\beta^{-1}}{L}\left[-\tr(\rho^{(1:L)}\beta(LF\mathbb{I}-\mathcal{H}^{(L)}))-\beta^{-1}S(\rho^{(1:L)})\right]\\
    =&\lim_{L\to\infty}\frac{1}{L} \left[\tr(\rho^{(1:L)}\mathcal{H}^{\otimes L})-S(\rho^{(1:L)})\right]-\F_{\text{eq}}~,
\end{split}
\end{equation}
where $\mathcal{H}^{(L)}\coloneqq\sum_{l}^L\bigotimes_{i=1}^{l-1}\mathds{1}_{Q_i}\otimes\mathcal{H}_{Q_l}\bigotimes_{i'=l+1}^L \mathds{1}_{Q_{i'}}$. For a degenerate Hamiltonian in the form of $\mathcal{H}=\omega\mathbb{I}/d$, the expression can be further simplified to
\begin{equation}
\label{eq:free_energy_rate}
    r = \beta^{-1}\left[\lim_{L\to\infty} \frac{-S(\rho^{(1:L)})}{L}+\ln d\right]~.
\end{equation}
So then what we really need to find is just the expression $\lim_{L\to\infty} \frac{-S(\rho^{(1:L)})}{L}$, which is essentially the entropy rate of the quantum system. Now we can likewise denote the combined systems of $Q_1\cdots Q_{L-1}$ as $B$ and $Q_L$ as $A$. Then, by the definition of quantum conditional entropy, we can obtain that
\begin{equation}
    S(\rho_{BA})=S(\rho_B)+S(A|B)_\rho~.
\end{equation}
Thus we conclude that in the limit $L\to\infty$
\begin{equation}
    s_{vN} = S(A|B)_\rho~.
\end{equation}
The von Neumann entropy rate is difficult to calculate; numerically, it is also difficult to compute since it involves eigen-decomposition of density matrices whose dimensions scale exponentially with $L$. We therefore turn to finding the upper bound of the entropy rate, which will then give us the lower bound on the free energy rate in Eq.~\eqref{eq:free_energy_rate}.

To find a meaningful upper bound on the entropy rate, we use the data-processing inequality, which states that quantum relative entropy cannot increase under a quantum channel. Here we can define a quantum channel that acts on $B$ while leaving $A$ untouched, i.e.,
\begin{equation}
    \eta_{AD} = (\text{id}_A\otimes \mathcal{E}_B)\rho_{AB} , \quad \eta_A\otimes\eta_D = (\text{id}_A\otimes \mathcal{E}_B)(\rho_A\otimes\rho_B)~.
\end{equation}
We can then show the quantum data processing inequality
\begin{equation}
    \begin{split}
        S(A|D)_{\eta} &= S(\eta_A)-\D(\eta_{AD}\|\eta_A\otimes\eta_D)\\
        &=S(\rho_A) - \D\left[ (\text{id}_A\otimes \mathcal{E}_B)(\rho_{AB})\| (\text{id}_A\otimes \mathcal{E}_B)(\rho_A\otimes\rho_B)\right]\\
        &\geq S(\rho_A)-\D (\rho_{AB}\| \rho_A\otimes\rho_B)\\
        &= S(A|B)_\rho
    \end{split}
\end{equation}
Any measurement on $B$ defines a valid quantum channel and therefore provides an upper bound on $S(A|B)$ through the data-processing inequality. We use the Helstrom measurement as one explicit choice; it need not be optimal for this entropy bound.

The Helstrom measurement minimizes the error probability when distinguishing two quantum states. Based on the structure of the quantum state, the most information an agent can obtain about the current latent state of the system comes from distinguishing the very last quantum state, which is either $\sigma^{(0)}$ indicating that the current latent state is in $S$, or $\sigma^{(1)}$, indicating that it is in $S'$. Here, $\tau_t^{(s)},s={0,1}$, denotes the $t$-partite quantum states ending with $\sigma^{(s)}$. 
In this case, the Helstrom measurement can be expressed as
\begin{equation}
    M_0 = \sum_{\lambda_i>0} \ket{\lambda_i}\!\bra{\lambda_i},\quad M_1 = \sum_{\lambda_i<0} \ket{\lambda_i}\!\bra{\lambda_i}~,
\end{equation}
where $p_0\tau_t^{(0)} -p_1\tau_t^{(1)}=\sum_i\lambda_i\ket{\lambda_i}\!\bra{\lambda_i}$ is the eigen-decomposition. 
The post-measurement state can then be written as
\begin{equation}
\begin{split}
    \eta_{AD}^{(t+1)}&=\sum_{s,y} \mu(s)\tr\left[M_y\tau_t^{(s)}\right] \ket{y}\!\bra{y}_D\otimes \xi_A^{(s)}\\
    &=\sum_yp_{y,t}\ket{y}\!\bra{y}_D\otimes {\chi_t^{(y)}}_A
\end{split}
\end{equation}
where $p_{y,t}=\sum_s\mu_0(s)\tr\left[M_y\tau_t^{(s)}\right]$ and $\chi_t^{(y)} = \frac{1}{p_{y,t}}\sum_s\mu_0(s)\tr(M_y\tau_t^{(s)})\xi^{(s)}$, and $\xi^{(s)}$ is the expected state conditioned on the latent state of the HMM, specifically,
\begin{equation}
    \xi^{(s)} = (1-p)\sigma^{(s)}+p\sigma^{(1-s)}~.
\end{equation}
The upper bound on the entropy rate can then be written as
\begin{equation}
    s_{vN} = S(A|B)_\rho \leq S(A|D)_\eta = \sum_y p_{y,t}S\left[\chi_t^{(y)}\right]
\end{equation}
which in turn provides us with the lower bound on the free energy rate,
\begin{equation}
    r \geq \beta^{-1}\lim_{t\to\infty}\left[ \ln d - S(A|D)_\eta\right]~,
\end{equation}
where the dependence on $t$ comes from $\tau_t^{(s)}$. For the numerical simulation, taking $t=12$ is enough to create a gap in performance.

\section{Derivation of causal dissipation}
\label{app:causal_dissipation}
We want to find the fundamental difference between a quantum agent that can extract work coherently across $L$ time steps and a classical agent that can only act locally at each step. 

We first consider a bipartite system and use this opportunity to restate the result from~\cite{zurek2003quantum,brodutch2010quantum} for a bipartite state
$\rho^{(1:2)}$. Assume that the agent measures in the time order, so from $Q_1\to Q_2$. For a classical demon who can operate sequentially and carry classical information, the extracted work can be quantified as 
\begin{equation}
\frac{W_C}{k_\mathrm{B}T} = \ln (d_{Q_1}d_{Q_2})-H(p_{\Pi_{Q_1}})-S(\rho_{Q_2}|\Pi_{Q_1})
\end{equation}
where $p_{\Pi_{Q_1}}$ is the probability distribution of measurement outcomes when applying the projector $\Pi_{Q_1}$ to the subsystem $Q_1$, while $H(X)$ is the Shannon entropy of the distribution of $X$. One can imagine measuring a quantum state of system $Q_1$ using $\Pi_{Q_1}$, collapsing the state to a pure state and extracting $\beta^{-1}\ln d_{Q_1}$ units of work, but having to spend $\beta^{-1}H(p_{\Pi_{Q_1}})$ units of work to reset the memory of the agent. The last term,
\begin{equation}
    S(\rho_{Q_2}|\Pi_{Q_1})=\sum_ip_{\Pi^{(i)}_{Q_1}} S(\rho_{Q_2|\Pi_{Q_1}^{(i)}})~,
\end{equation}
represents the uncertainty in $Q_2$ after the agent measures $Q_1$, $\rho_{Q_2|\Pi_{Q_1}^{(i)}}$ is the normalized post-measurement quantum state after obtaining result $i$ in the first measurement on $Q_1$, i.e.,
\begin{equation}
    \rho_{Q_2|\Pi_{Q_1}^{(i)}} = \frac{\Pi^{(i)}_{Q_1}\rho^{(1:2)}\Pi^{(i)}_{Q_1}}{\tr\left[\Pi^{(i)}_{Q_1}\rho^{(1:2)}\right]}~.
\end{equation}
Note that we omit the explicit tensor product expression for the projector to reduce clutter, but in principle, the full projector is $\Pi^{(i)}_{Q_1}\otimes\mathbb{1}_{Q_2}$.

On the other hand, a quantum agent will quite straightforwardly extract 
\begin{equation}
    \frac{W_Q}{k_\mathrm{B}T} = \ln (d_{Q_1}d_{Q_2})-S(\rho^{(1:2)})~.
\end{equation}
The difference between the classical and quantum demon is then given by 
\begin{equation}
    \frac{W_Q-W_C}{k_\mathrm{B}T}=H(p_{\Pi_{Q_1}})+S(\rho_{Q_2}|\Pi_{Q_1})-S(\rho^{(1:2)})~.
\end{equation}

One can then optimize the measurement $\Pi_{Q_1}$ to minimize this difference, which then gives us the expression for \emph{thermal discord}.
\begin{equation}
    \delta(Q_2|Q_1) \coloneqq \min_{\Pi_{Q_1}}\left[H(p_{\Pi_{Q_1}})+S(\rho_{Q_2}|\Pi_{Q_1})\right]-S(\rho^{(1:2)})~.
\end{equation}
We can generalize this quantity to a $L$-partite case and obtain the causal dissipation, which quantifies the fundamental lower bound on the difference between the work extracted by a fully quantum agent and an agent with only classical memory. 

Suppose the agent is given the multi-time state $\rho^{(1:L)}$ in Eq.~\eqref{eq:fuel_state}. The agent has to operate in temporal order, i.e., $Q_1\to Q_2\to\cdots\to Q_L$. At time step $t$, the agent is allowed to vary its measurement basis conditioned on the history of previous measurements and outcomes, $\tilde{\mathbf{H}}_{t-1}\coloneqq(\Pi_{Q_1},O_1,\cdots,\Pi_{Q_t-1},O_{t-1})$. It carries out measurements in the form of $\Pi_{Q_t}=\Pi_{Q_t|\tilde{\mathbf{H}}_{t-1}}$.  The causal dissipation can then be defined as
\begin{equation}
\label{eq:causal_def}
    \beta\delta(Q_{\overrightarrow{1:L}}) \coloneqq \min_{\overrightarrow{\gbm\Pi}} \mathbb{E}_{\tilde{\mathbf{H}}_L|\overrightarrow{\gbm\Pi}} \left[ \sum_{t=1}^{L-1} H(O_t | \tilde{\mathbf{H}}_{t-1}) + S(\tilde{\rho}_{Q_L | \tilde{\mathbf{H}}_{L-1}}) \right] - S(\rho^{(1:L)})~.
\end{equation}
where $\gbm{\overrightarrow{\Pi}}\coloneqq\{\Pi_{Q_t|\tilde{\mathbf{H}}_{t-1}}\}_{i=1}^L$ is the sequence of measurements applied to subsystem $Q_t$ conditioned on histories, $\tilde{\mathbf{H}}_{t-1}$. 
The conditional Shannon entropies account for the increase in the agent's memory due to the measurement, with $H(p_{\Pi_{A_1}}) \coloneqq -\sum_{o_1}p^{(o_1)}\ln p^{(o_1)}$.

\subsection{Equivalence between work deficit and causal dissipation}
\label{app:prove_dissipation}
In this subsection, we will prove Theorem~\ref{thm:thm3}. We first clarify some notational differences. In Eq.~\eqref{eq:causal_def}, the causal dissipation is defined with respect to $\tilde{\mathbf{H}}_{L}=(\Pi_{Q_1},O_1,\cdots,\Pi_{Q_L},O_{L})$, while in Theorem~\ref{thm:thm3}, it is defined with respect to $\mathbf{H}_{L}=(A_1,W_1,\cdots,A_L,W_{L})$. The reason for this, as we will show in this proof, is that the actions taken $A_t$ are mathematically equivalent to performing the measurement $\Pi_{Q_t}$, and as a result the distribution of $W_t$ is identical to that of $O_t$ as shown in Eq.~\eqref{eq:work_dist}. See Corollary~\ref{corrollary:optimal_eig} for more information. Hence, the substitution is straightforward.
\begin{proof}
    Let $\tilde{\rho}_{Q_t | \mathbf{h}_{t-1}}$ be the reduced local state of the subsystem $Q_t$ conditioned on the observable history of applied protocols and extracted work $\mathbf{h}_{t-1}$. At each time step $t$, the optimal policy $\Lambda^*$ dictates a tailored target state $\rho_{\Lambda^*(t)}$. The maximum cumulative expected work extracted or TOFE can be expanded as
    \begin{equation}
        \F_{\text{TO}}^{(1:L)}(\rho^{(1:L)}) =\E_{\mathbf{H}_L|\Lambda^*} \left[ \sum_{t=1}^{L} \beta^{-1} \left( \D(\tilde{\rho}_{Q_t | \mathbf{H}_{t-1}} \| \gamma) - \D(\tilde{\rho}_{Q_t | \mathbf{H}_{t-1}} \| \rho_{\Lambda^*(t)}) \right) \right]~.
    \end{equation}
We expand the relative entropy terms using the identity $\D(\rho \| \sigma) = -S(\rho) - \tr(\rho \ln \sigma)$ and replace $\gamma=\mathbb {I}/d$ for degenerate Hamiltonian. This yields
\begin{equation}
    \D(\tilde{\rho}_{Q_t | \mathbf{H}_{t-1}} \| \gamma) - \D(\tilde{\rho}_{Q_t | \mathbf{H}_{t-1}} \| \rho_{\Lambda^*(t)}) = \ln d + \tr(\tilde{\rho}_{Q_t | \mathbf{H}_{t-1}} \ln \rho_{\Lambda^*(t)})~.
\end{equation}
Substituting this back into the TOFE, we get
\begin{equation}
    \beta \mathcal{F}_{\text{TO}}^{(1:L)}(\rho^{(1:L)}) = L \ln d + \mathbb{E}_{\mathbf{H}_L|\Lambda^*} \left[ \sum_{t=1}^{L} \tr(\tilde{\rho}_{Q_t | \mathbf{H}_{t-1}} \ln \rho_{\Lambda^*(t)}) \right]~.
\end{equation}
\begin{corollary}[Informational Equivalence of Optimal Protocols]
\label{corrollary:optimal_eig}
    As a direct consequence of Theorem~\ref{thm2}, the optimal tailored target state $\rho_{\Lambda^*(t)}$ establishes a mathematical correspondence between thermodynamic work extraction and the Shannon entropy of the classical measurement outcomes. Specifically, the trace evaluates as:
    \begin{equation}
        \tr(\tilde{\rho}_{Q_t | \mathbf{h}_{t-1}} \ln \rho_{\Lambda^*(t)}) = -H(W_t | \mathbf{h}_{t-1})~.
    \end{equation}
    $H(X)$ here refers to the Shannon entropy of the probability distribution of $X$.
\end{corollary}
\begin{proof}
     By Theorem~\ref{thm2}, to minimize the local dissipation, the optimal tailored state is given by 
     \begin{equation}
         \rho_{\Lambda^*(t)} = \sum_{i}\lambda_i^*\ket{\psi_i}\!\bra{\psi_i}~,
     \end{equation}
     where the eigenvalues $\lambda_i^*$ must perfectly match the Born-rule probabilities of the actual work outcomes:
     \begin{equation}
         \lambda_i^* = \langle \psi_i | \tilde{\rho}_{Q_t | \mathbf{h}_{t-1}} | \psi_i \rangle = \Pr(W_t = w_i | \mathbf{h}_{t-1})~.
     \end{equation}
     Because the logarithm of this optimal state is $\ln \rho_{\Lambda^*(t)} = \sum_i \ln \lambda_i^*|\psi_i\rangle\langle\psi_i|$, evaluating the trace against the physical post-measured state yields:
\begin{equation}
    \tr(\tilde{\rho}_{Q_t | \mathbf{h}_{t-1}} \ln \rho_{\Lambda^*(t)}) = \sum_i \lambda_i^* \ln \lambda_i^* = -H(W_t | \mathbf{h}_{t-1}) ~.
\end{equation}
\end{proof}
By Corollary~\ref{corrollary:optimal_eig}, the trace term in our expanded expression maps exactly to the negative Shannon entropy of the extracted work outcomes conditioned on the history. Substituting this identity, the TOFE can be rewritten purely in terms of the information entropy generated along the sequence:
\begin{equation}
    \beta \mathcal{F}_{\text{TO}}^{(1:L)}(\rho^{(1:L)}) = L \ln d - \mathbb{E}_{\mathbf{H}_{L}|\Lambda^*} \left[ \sum_{t=1}^{L} H(W_t | \mathbf{H}_{t-1}) \right]~.
\end{equation}

At the final time step $L$, the agent extracts work without needing to acquire predictive information for future steps. Consequently, the optimal policy $\Lambda^*(L)$ defaults to a greedy strategy, perfectly aligning the tailored state with the final physical state $\tilde{\rho}_{Q_L| \mathbf{h}_{L-1}}$. Because of this, the Shannon entropy of the classical measurement outcome becomes mathematically identical to the von Neumann entropy of the conditional quantum state,
\begin{equation}
    H(W_L | \mathbf{h}_{L-1}) = S(\tilde{\rho}_{Q_L | \mathbf{h}_{L-1}})~.
\end{equation}
This then splits the summation to give
\begin{equation}
\label{eq:final_TO}
    \beta \mathcal{F}_{\text{TO}}^{(1:L)} = L \ln d - \mathbb{E}_{\mathbf{H}_{L}|\Lambda^*} \left[ \sum_{t=1}^{L-1} H(W_t | \mathbf{H}_{t-1}) + S(\tilde{\rho}_{Q_L | \mathbf{H}_{L-1}}) \right]~.
\end{equation}
Next, the non-equilibrium free energy of the global multi-time state is given by
\begin{equation}
    \beta \mathcal{F}_{\text{noneq}}^{(1:L)} = L \ln d - S(\rho^{(1:L)})~.
\end{equation}
The expectation term in our TOFE expansion represents the total entropy of the system after sequential interventions. By definition, causal dissipation $\delta(Q_{\overrightarrow{1:L}})$ is the minimal difference between this intervened entropy and the joint entropy $S(\rho^{(1:L)})$:
\begin{equation}
    \delta(Q_{\overrightarrow{1:L}})= \beta^{-1}\mathbb{E}_{\mathbf{H}_{L}|\Lambda^*} \left[ \sum_{t=1}^{L-1} H(W_t | \mathbf{H}_{t-1}) + S(\tilde{\rho}_{Q_L | \mathbf{H}_{L-1}}) \right] - \beta^{-1}S(\rho^{(1:L)})~.
\end{equation}
Rearranging to isolate the expectation and substituting this expression back into Eq.~\eqref{eq:final_TO} proves the theorem.
\begin{equation}
\begin{split}
    \beta \mathcal{F}_{\text{TO}}^{(1:L)} &= L \ln d - \left( \beta\delta(Q_{\overrightarrow{1:L}}) + S(\rho^{(1:L)}) \right)\\
    &=\left( L \ln d - S(\rho^{(1:L)}) \right) - \beta\delta(Q_{\overrightarrow{1:L}})\\
    \mathcal{F}_{\text{TO}}^{(1:L)} &=  \mathcal{F}_{\text{noneq}}^{(1:L)} - \delta(Q_{\overrightarrow{1:L}})~.
\end{split}
\end{equation}
\end{proof}

\subsection{Mathematical properties of causal dissipation}
\label{sec:dissipation_property}
Here, we try to establish some mathematical properties of causal dissipation.

\begin{enumerate}
 \item  Positivity: \label{property:positivity}
 $\delta(Q_{\overrightarrow{1:L}})\geq0$. The sequence of adaptive measurements/operations can be modeled as a single global unital CPTP dephasing channel $\Phi$ acting on the first $L-1$ subsystems. The von Neumann entropy of the resulting block-diagonal state evaluates exactly to the Shannon entropy of the classical trajectory plus the residual quantum entropy:
 \begin{equation}
     S(\Phi(\rho^{(1:L)})) = H(O_{1:L-1}) + \mathbb{E}_{\mathbf{H}_{1:L-1}}~S(\tilde{\rho}_{Q_L | \mathbf{H}_{1:L-1}})~.
 \end{equation}
 Using the classical chain rule $H(O_{1:L-1}) = \sum_{t=1}^{L-1} H(O_t | O_{1:t-1})$, these are precisely the classical entropies in the definition of causal dissipation. Because unital channels cannot decrease von Neumann entropy (by Klein's inequality~\cite{nielsen2010quantum}), $S(\Phi(\rho^{(1:L)})) \ge S(\rho^{(1:L)})$, which guarantees $\delta(Q_{\overrightarrow{1:L}}) \geq 0$.
    \item \label{property:i.i.d}$\delta(Q_{\overrightarrow{1:L}})=0$ when the subsystems are not correlated, i.e., 
        \begin{equation}
            \rho^{(1:L)} = \rho^{\otimes L}~,
        \end{equation}
        then, by definition of causal dissipation, the measurements on the preceding systems will not affect the state of the subsequent systems; hence the causal dissipation will be $0$. For the perturbed coin process, this can be seen when $p=1/2$, for which the state takes the form of $\rho^{(1:L)} = \left[\frac{1}{2}(\sigma^{(0)}+\sigma^{(1)})\right]^{\otimes L}$. 

    \item \label{property:CQ}$\delta(Q_{\overrightarrow{1:L}})=0$ if the multipartite state is of the form 
        \begin{equation}
            \rho^{(1:L)} = \sum_{i_1\cdots i_{L}} \Pr(i_1\cdots i_{L})\bigotimes_{t=1}^{L-1}\ket{i_t}\!\bra{i_t}\otimes \rho^{(i_L)}_{Q_L}~,
        \end{equation}
        where $\ket{i_t}$ form an orthonormal basis. The agent can always measure in the orthonormal eigenbasis of the first $L-1$ systems without disturbing the succeeding system. Note that this can be viewed as a stricter constraint compared to the bipartite analogue of a classical-quantum state. This property also demonstrates the \emph{asymmetry} of causal dissipation, i.e., $\delta(Q_{\overrightarrow{1:L}})\neq\delta(Q_{\overleftarrow{1:L}})$.
    \item \label{property:rate}$\delta(Q_{\overrightarrow{1:L}})\neq0 \nRightarrow \lim_{L\to\infty}\frac{1}{L}\delta(Q_{\overrightarrow{1:L}})\neq0$, i.e., the \emph{rate} of causal dissipation can asymptotically approach 0, even though the cumulative dissipation over a finite number of time steps is non-zero. To illustrate this, we consider the state produced by the perturbed coin when $p=0$; the quantum state will take the following form:
    \begin{equation}
        \rho^{(1:L)} = \frac{1}{2}\sigma^{(0)\otimes L} + \frac{1}{2}\sigma^{(1)\otimes L},
    \end{equation}
    where $\sigma^{(0)}$ and $\sigma^{(1)}$ are not orthogonal. Clearly, if we just consider the case where $L$ is not too large, the causal dissipation would not be 0, since $\sigma^{(0)}$ and $\sigma^{(1)}$ do not commute; any measurement will disturb the systems after the measurements. Conceptually, causal dissipation arises precisely because measurement on one subsystem produces indistinguishable conditional states on the other subsystems. Unlike the finite-$L$ case, in the asymptotic limit, it is always possible for the agent to distinguish the 2 states, $\sigma^{(0)\otimes L}$ or $\sigma^{(1)\otimes L}$. Further measurements no longer provide useful information, so the rate of causal dissipation approaches 0 as $L\to\infty$. 
\end{enumerate}

\begin{proof}
Formally, we wish to demonstrate that
\begin{equation}
    \lim_{L\to\infty} \frac{1}{L}\delta(Q_{\overrightarrow{1:L}}) = 0~,
\end{equation}
even though causal dissipation over any finite number of time steps is clearly non-zero. We can rearrange the terms to obtain that
\begin{equation}
    \lim_{L\rightarrow\infty} \frac{1}{L} \min_{\overrightarrow{\gbm\Pi}} \mathbb{E}_{\tilde{\mathbf{H}}_L|\overrightarrow{\gbm\Pi}}\left[\sum_{t=1}^{L-1} H(O_t|\tilde{\mathbf{H}}_{t-1}) + S(\tilde{\rho}_{Q_L|\tilde{\mathbf{H}}_{L-1}})\right] = \lim_{L\rightarrow\infty} \frac{1}{L} S(\rho^{(1:L)})~,
\end{equation}
where the RHS is the entropy rate of the quantum state.
Here we consider the cases in which the HMM has $p=0,\text{ or }1$ to illustrate. 
First, we realize that in both cases where $p=0$ or $p=1$, the total quantum state can be written as 
\begin{equation}
    \rho^{(1:L)} = \frac{1}{2}\rho_0 + \frac{1}{2}\rho_1~,
\end{equation}
where $\rho_0,\rho_1 = \sigma^{(0)\otimes L},\sigma^{(1)\otimes L}$ when $p=0$  and $\rho_0,\rho_1 = (\sigma^{(0)}\otimes\sigma^{(1)})^{\otimes L/2}, (\sigma^{(1)}\otimes\sigma^{(0)})^{\otimes L/2}$ when $p=1$.

We first prove that the actual entropy rate of the total quantum state tends to 0.
Here, we consider a quantity 
    \begin{equation}
        \chi_L = S(\rho^{(1:L)})-\frac{1}{2}S(\rho_0)-\frac{1}{2}S(\rho_1)~.
    \end{equation}
This quantity is known as the Holevo information, or Holevo quantity,  which is upper-bounded by the classical description of the ensemble, i.e., $0\leq \chi_L\leq H_2(\frac{1}{2})$ since both states appear with probability $1/2$ ; here, $H_2$ is the binary entropy. 
We then rewrite the entropy rate with this quantity.
\begin{equation}
    h \coloneqq \lim_{L\to\infty}\frac{1}{L}S(\rho^{(1:L)}) = \lim_{L\to\infty}\frac{\chi_L}{L} + \frac{1}{2L}S(\rho_0)+\frac{1}{2L}S(\rho_1)~.
\end{equation}
Since the Holevo information is bounded, this results in
\begin{equation}
    h = \lim_{L\to\infty}\frac{1}{2L}S(\rho_0)+\frac{1}{2L}S(\rho_1)~.
\end{equation}
For our case, this admits a closed-form solution and equals $0$ when both $\sigma^{(0)}$ and $\sigma^{(1)}$ are pure. 

Now we have to prove that 
\begin{equation}
     \lim_{L\rightarrow\infty} \min_{\overrightarrow{\gbm\Pi}} \frac{1}{L} \mathbb{E}_{\tilde{\mathbf{H}}_L|\overrightarrow{\gbm\Pi}}\left[\sum_{t=1}^{L-1} H(O_t|\tilde{\mathbf{H}}_{t-1}) + S(\tilde{\rho}_{Q_L|\tilde{\mathbf{H}}_{L-1}})\right]  =0
\end{equation}
Both the sum of Shannon entropies and the final von Neumann entropy are non-negative. It therefore suffices to find a measurement sequence, $\overrightarrow{\gbm\Pi}$, for which both terms converge to zero. 

From results in symmetric state discrimination, given two pure quantum states $\rho$ and $\sigma$ with equal prior probabilities, one may distinguish them asymptotically. In particular, there exists a sequence of measurements $\{\Pi_1^{(i_1)},\Pi_2^{(i_2)},\ldots,\Pi_L^{(i_L)}\}$ such that
\begin{equation}
\lim_{L\to\infty}\tr(\Pi_1^{(i_1)}\otimes\Pi_2^{(i_2)}\otimes\ldots\otimes\Pi_L^{(i_L)} \rho^{\otimes L}) = 1 
\end{equation}
and 
\begin{equation}
\lim_{L\to\infty}\tr(\Pi_1^{(i_1)}\otimes\Pi_2^{(i_2)}\otimes\ldots\otimes\Pi_L^{(i_L)} \sigma^{\otimes L}) = 0 ~.
\end{equation}
The agent's strategy $\vec{\Pi}^*$ is as follows: measure in the eigenbasis of $\rho$, such that $\tr(\Pi\rho) = 1$ and $\tr(\Pi\sigma)=r$. As long as the measurement result corresponds to $\Pi$, the agent continues measuring in that basis; on the other hand, if the agent observes an outcome corresponding to $\mathds{1}-\Pi$, it chooses to measure in the eigenbasis of $\sigma$, i.e., $\Pi'$ such that $\tr (\Pi'\sigma)=1$, $\tr(\Pi'\rho)=r$. 

We evaluate the expected classical memory entropy, $\mathbb{E}[H(O_t|\tilde{\mathbf{H}}_{t-1})]$. Suppose the agent observes the outcome corresponding to $\Pi$ for the first $k$ consecutive times. The post-measurement global state is updated to:

\begin{equation}
\label{eq:post_state}
\begin{split}
    \tilde\rho^{(1:L)} = \frac{(\Pi^{\otimes k}\otimes \mathds{1}^{\otimes(L-k)})\rho^{(1:L)})(\Pi^{\otimes k}\otimes \mathds{1}^{\otimes(L-k)})}{\tr((\Pi^{\otimes k}\otimes \mathds{1}^{\otimes(L-k)})\rho^{(1:L)})} &= \frac{\frac{1}{2}\rho^{\otimes L}+\frac{1}{2}r^k\rho^{\otimes L}\otimes \sigma^{\otimes (L-k)}}{\frac{1}{2}+\frac{1}{2}r^k}\\
    &= \frac{1}{1+r^k} \rho^{\otimes L} + \frac{r^k}{1+r^k}\rho^{\otimes k}\otimes \sigma^{\otimes (L-k)}
\end{split}
\end{equation}
For the $(k+1)$-th measurement, the probability of obtaining the $\Pi$ outcome again is $\frac{1+r^{k+1}}{1+r^k}$. The conditional entropy of this specific outcome is given by the binary entropy: $H_2\left(\frac{1+r^{k+1}}{1+r^k}\right)$.

If, conversely, the agent had previously observed an anomaly (a $\mathds{1}-\Pi$ outcome), the state collapses entirely to the $\sigma$ branch. The agent switches to $\Pi'$, making all subsequent outcomes deterministic, yielding a binary entropy of $H_2(1) = 0$.
Taking the expectation over the history space $\tilde{\mathbf{H}}_{k}$, the expected classical entropy at step $k+1$ is:
\begin{equation}
    \mathbb{E}_{\tilde{\mathbf{H}}_k}[H(O_{k+1}|\tilde{\mathbf{H}}_k)] = \frac{1+r^k}{2} H_2\left(\frac{1+r^{k+1}}{1+r^k}\right) + \frac{1-r^k}{2} H_2(1)
\end{equation}

Because $0 \le r < 1$, the ratio $\frac{1+r^{k+1}}{1+r^k} \rightarrow 1$ as $k\to\infty$. By continuity of the binary entropy, the corresponding conditional entropy therefore approaches $H_2(1) = 0$. This means the expected classical entropy of the measurement sequence converges to $0$,
\begin{equation}
    \lim_{k\rightarrow\infty} \mathbb{E}_{\tilde{\mathbf{H}}_k}[H(O_{k+1}|{\tilde{\mathbf{H}}_k})] = 0~.
\end{equation}
We use the property of Ces$\grave{a}$ro limits~\cite{rudin2021principles}, which states that if a sequence converges to a limit $L$, then the sequence of its averages also converges to $L$.
\begin{equation}
    \lim_{L\rightarrow\infty} \frac{1}{L} \mathbb{E}_{\tilde{\mathbf{H}}_L}\sum_{t=1}^{L-1} [H(O_t|\tilde{\mathbf{H}}_{t-1})] = 0~.
\end{equation}

Finally, we evaluate the expected von Neumann entropy of the final conditional subsystem. Successive measurements can distinguish wh    ether the underlying state is $\rho^{\otimes L}$ or $\sigma^{\otimes L}$ with an error probability that vanishes as $L\to\infty$. The final conditional state $\tilde{\rho}_{Q_L|\tilde{\mathbf{H}}_{L-1}}$ will thereby be purified to either $\rho$ or $\sigma$. Because both are pure states, their von Neumann entropies are zero.
\begin{equation}
    \lim_{L\rightarrow\infty} \frac{1}{L} \mathbb{E}_{\tilde{\mathbf{H}}_L}[S(\tilde{\rho}_{Q_L|\tilde{\mathbf{H}}_{L-1}})] = 0~.
\end{equation}

Since both the measurement cost and the final quantum uncertainty asymptotically vanish under this specific measurement policy $\vec{\Pi}^*$, the minimal expectation over all policies must also equal zero. Therefore:
\begin{equation}
    \lim_{L\rightarrow\infty} \frac{1}{L} \delta(Q_{\overrightarrow{1:L}}) =0~.
\end{equation}
This completes the proof.
\end{proof}

\subsection{Numerical Simulation}
\label{sec:numerical}

In Fig.~\ref{fig:full_compare}, we demonstrate the agreement between the simulated work deficit and the causal dissipation for both $L=3$ and $L=4$. Note that the lack of comparison for a longer time horizon is strictly a restriction due to the brute force optimization in calculating the causal dissipation, which scales exponentially with $L$ rather than being a limitation of our dynamic programming, which scales linearly with $L$. 
We use the following parameters and Alg.~\ref{alg:DPP_for_work} to find the optimal policy for $p\in (0,1)$ and $r\in(0,1)$.
\begin{enumerate}
    \item The emitted quantum state can take on either $\ket{\phi_0}=\ket{0}$ or $\ket{\phi_1}=\sqrt{r}\ket{0}+\sqrt{1-r}\ket{1}$.
    \item Belief states: $\mathcal{K}= \left\{ \left( 1/2 + \epsilon, 1/2 - \epsilon \right) \ \middle|\ \epsilon \in \left[-1/2, 1/2\right] \right\}$;  we divide the parameter $\epsilon$ into 300 equal-sized intervals.
    \item Actions: $\A$ consists of all bases spanned by $\ket{\phi_0}$ and $\ket{\phi_1}$, taking the form of $\A=\{\ket{\psi_{\theta}}=\cos\frac{\theta}{2}\ket{0}+\sin\frac{\theta}{2}\ket{1}|\theta\in[0,2\pi)\}$; we also divide the parameter $\theta$ into 300 equal sized intervals. 
\end{enumerate}

After the optimal policy is found, we apply it by initializing the agent at belief state $K_0=\mu_0 = \gbm\pi=(1/2,1/2)$ of the perturbed coin process in Fig.~\ref{fig:combined_fig}(a). The agent then finds the optimal action based on the policy provided at each time step, conditioned on its belief state. The resulting work extracted is deducted from the non-equilibrium free energy of the multi-time state shown in Eq.~\eqref{eq:fuel_state}. This is then averaged across $100$ repetitions for both $L=3$ and $L=4$. The optimization of causal dissipation is done using \texttt{scipy.optimize.differential\_evolution}.

\begin{figure}
    \centering
    \includegraphics[width=0.7\linewidth]{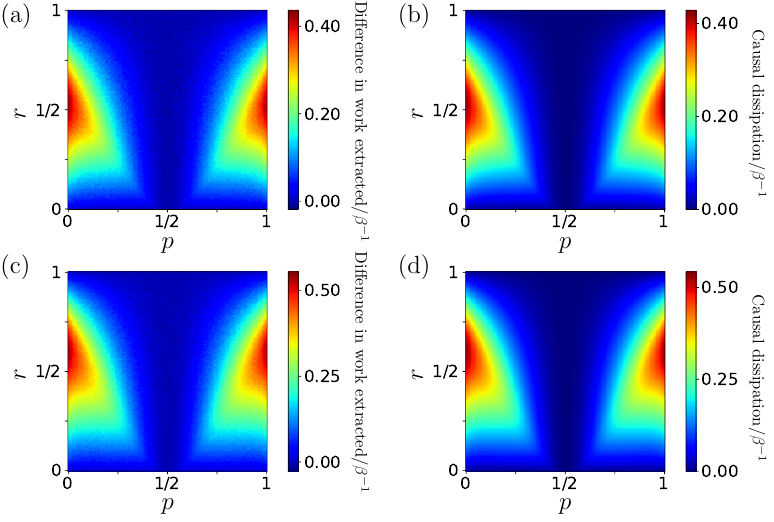}
    \caption{The gap between the non-equilibrium free energy and the TOFE (panels (a) and (c)) shows good agreement with the analytical causal dissipation (panels (b) and (d)). The top row corresponds to a time horizon of $L=3$, while the bottom row is for $L=4$.}
    \label{fig:full_compare}
\end{figure}

\section{Bounding the Cost of Agent Memory Updates}
\label{sec:mem_cost}
In this section, we investigate the thermodynamic cost associated with updating the agent's memory during work extraction. We demonstrate that memory updates can be performed with negligible additional overhead. 

We begin by observing that at any time step $t$, the agent has an internal memory that keeps track of the belief state $\gbm\eta_{t-1}$, and the action taken by the agent is determined by a policy $A_t =\Lambda(\gbm\eta_{t-1})$. The agent then executes the action on the quantum state $\sigma_{Q_t}$ and obtains work $W_t$ after doing a projective measurement on the battery. The agent then updates its internal memory to $\gbm\eta_t$ using an update rule $\tau$. For this paper, we use Bayesian inference for $\tau$.

\begin{equation}
\begin{split}
\gbm\eta_{t} &= \tau(\gbm\eta_{t-1}, w_t) \\
&= N_{t-1}^{-1} \sum_{x_{t}\in\mathcal{X}} \Pr(W_{t}=w_{t}|x_{t},a_{t},\gbm\eta_{t-1}) \gbm\eta_{t-1} \mathsf{T}^{(x_t)}~.
\end{split}
\end{equation} 

Throughout our discussion, the Hamiltonian is degenerate; hence the work distribution is given by 
\begin{equation}
\label{eq:pdf_work}
    \Pr(W=w_i) = \bra{\lambda_i}\sigma\ket{\lambda_i},\hquad\beta w_i \coloneqq \ln{\lambda_i}+\ln d~,
\end{equation} 
where $\lambda_i$ is the $i$-th eigenvalue of the state to which the action is tailored and $\ket{\lambda_i}$ is its corresponding eigenvector. 

We first establish the necessary condition for a dissipation-free update. For each time $t$, we define a supported update map
\begin{equation}
    F_t:(\gbm\eta_{t-1},w_t) \mapsto (\tau(\gbm\eta_{t-1},w_t),w_t)~,
\end{equation}
whose domain contains only pairs $(\gbm\eta,w)$ that can occur with non-zero probability under the optimal policy.
\begin{prop}[Reversibility criterion]
    The update can be implemented by a unitary on the
memory and the current battery, without an additional ancilla, if $F_t$ is injective on its supported domain. 
\end{prop}
We now prove the following lemma regarding this implementation.
\begin{lemma}[Injective updates admit controlled-unitary implementations]
\label{lemma:injective_unitary}
    Suppose that for a fixed time $t$, the map $F_t$ 
    \begin{equation}
        F_t:(\gbm\eta_{t-1},w_t) \mapsto (\tau(\gbm\eta_{t-1},w_t),w_t)~,
    \end{equation}
    is injective on all supported pairs. Then there exists a unitary $U_t$ on $M\otimes B_t$ satisfying
    \begin{equation}
    \label{eq:batt_control_unitary}
        U_t\ket{\gbm\eta_{t-1}}_M\ket{w_t}_{B_t} = \ket{\tau(\gbm\eta_{t-1},w_t)_M}\ket{w_t}_{B_t}
    \end{equation}
    for every such supported pair. The unitary can be chosen to commute with $\mathcal{H}_M+\mathcal{H}_{B_t}$.
\end{lemma}
\begin{proof}
Different supported input pairs correspond to orthogonal states. Injectivity ensures that their proposed output states are also different and therefore orthogonal. The map in Eq.~\eqref{eq:batt_control_unitary} is consequently an isometry on the supported subspace. Any finite-dimensional isometry between subspaces of equal dimension can be completed to a unitary on the full space.

More explicitly, for each battery value $w$, injectivity says that $\gbm\eta \mapsto \tau(\gbm\eta,w)$ is one-to-one on the beliefs compatible with $w$. It can therefore be extended to a permutation, and hence to a unitary $U_{t,w}$ of
the whole memory basis. Define
\begin{equation}
    U_t = \sum_{w}\ket{w}\!\bra{w}_{B_t}\otimes U_{t,w}~,
\end{equation}
with an arbitrary unitary extension on the battery. This operation does not alter $w$, and $U_{t,w}$ acts within the degenerate memory. Hence it commutes with $\mathcal{H}_M+\mathcal{H}_{B_t}$.
\end{proof}

Thus, if injectivity holds, we are guaranteed to have zero dissipation. However, if it does not hold directly, we demonstrate below that it can be restored via spectral tagging with an arbitrarily small increase in energy dissipation.

\subsection{Spectral tagging}
Let $\xi_{\gbm\eta}$ be the expected state conditioned on the belief $\gbm\eta$. For a chosen basis, define 
\begin{equation}
    p_{t,\gbm\eta,i} \coloneqq \bra{\psi_{t,\gbm\eta,i}}\xi_{\gbm\eta}\ket{\psi_{t,\gbm\eta,i}}~.
\end{equation}
The optimal target spectrum for this fixed basis is $p_{t,\gbm\eta}$. The unperturbed optimal target state thus takes the form: 
\begin{equation}
    \rho_{t,\gbm\eta}=\sum_{i=1}^d p_{t,\gbm\eta,i}\ket{\psi_{t,\gbm\eta,i}}\!\bra{\psi_{t,\gbm\eta,i}}~.
\end{equation}

We introduce a tagging strength 0 < $\alpha$ < 1. For every $(t, \gbm\eta)$, choose an auxiliary probability vector
\begin{equation}
    r_{t,\gbm\eta} = (r_{t,\gbm\eta,1}, \cdots , r_{t,\gbm\eta,d})
\end{equation}
located strictly in the interior of the probability simplex, ensuring every component is strictly positive. We define the perturbed spectrum as
\begin{equation}
\label{eq:perturb_spec}
    q_{t,\gbm\eta,i} = (1 - \alpha)p_{t,\gbm\eta,i} + \alpha r_{t,\gbm\eta,i}~.
\end{equation}
\begin{lemma}[Existence of non-colliding tags]
\label{lemma:non-colliding}
    For every fixed $t$, the auxiliary vectors $r_{t,\gbm\eta}$ can be chosen so that all numbers
    \begin{eqnarray}
        \{q_{t,\gbm\eta,i}:\gbm\eta\in\mathcal{K}_{t-1}, i=1,\cdots,d\}
    \end{eqnarray}
    are pairwise distinct. Furthermore, every $q_{t,\gbm\eta,i}$ is then strictly positive.
\end{lemma}
\begin{proof}
    There are only finitely many beliefs and indices. For each distinct pair $(\gbm\eta, i) \neq (\gbm\eta', j)$,
the collision equality
\begin{equation}
    (1-\alpha)p_{t,\gbm\eta,i}+\alpha r_{t,\gbm\eta,i}=(1-\alpha)p_{t,\gbm\eta',j}+\alpha r_{t,\gbm\eta',j}
\end{equation}
defines a proper affine hyperplane within the product of the auxiliary probability simplices. Since a finite union of proper hyperplanes cannot cover a non-empty open set, we may choose the collection $\{r_{t,\gbm\eta}\}_{\gbm\eta}$ in the complement of all these hyperplanes. Because each $r_{t,\gbm\eta}$ is chosen to be in the interior and $\alpha>0$, every component of $q_{t,\gbm\eta}$ is therefore positive even when the corresponding $p_{t,\gbm\eta}=0$. This property also ensures the local dissipation term remains finite when the original tailored state is not full-rank.
\end{proof}
We construct the tagged target state 
\begin{equation}
    \tilde \rho_{t,\gbm\eta} = \sum_{i=1}^d q_{t,\gbm\eta,i} \ket{\psi_{t,\gbm\eta,i}}\!\bra{\psi_{t,\gbm\eta,i}}~,
\end{equation}
which yields the corresponding tagged work values: 
\begin{equation}
\label{eq:new_work}
    \beta \tilde w_{t,\gbm\eta,i} = \ln q_{t,\gbm\eta,i} + \ln d~.
\end{equation}
Since the logarithm is a one-to-one function, Lemma~\ref{lemma:non-colliding} guarantees that these perturbed work values uniquely identify the specific pair $(\gbm\eta,i)$ at any time step. Beyond injectivity, we must verify that this spectral tagging does not alter the underlying Bayesian inference dynamics.

\begin{lemma}[Spectral tagging preserves the complete observation process]
\label{lemma:spectral_preserve}
    Assume that the observed unperturbed work values of $\rho_{t,\gbm\eta}$ are pairwise distinct (if this assumption does not hold, please refer to Appendix~\ref{sec:local_perturb}). Under this assumption, replacing $\rho_{t,\gbm\eta}$ with $\tilde \rho_{t,\gbm\eta}$ leaves the following unchanged:
    \begin{enumerate}
        \item the probability of every branch conditioned on every emitted state $\sigma^{(x)}$;
        \item the Bayesian posterior conditioned on that branch;
        \item the probability distribution over future beliefs;
        \item the future bases selected by the original policy.
    \end{enumerate}
\end{lemma}
\begin{proof}The perturbed and unperturbed target states share exactly the same eigenbasis. Branch probabilities depend exclusively on this basis, independently of the target's eigenvalues, as shown in Eq.~\eqref{eq:pdf_work}. Therefore, the likelihoods utilized in Bayes' rule remain identical, proving points 1 and 2.

    While the numerical work value changes, Eq.~\eqref{eq:new_work} provides the one-to-one decoding from $\tilde w_{t,\gbm\eta,i} \leftrightarrow(\gbm\eta,i)$. We can therefore define the tagged update
    \begin{equation}
        \tilde\tau(\gbm\eta,\tilde w_{t,\gbm\eta,i})\coloneqq \tau(\gbm\eta,i)~.
    \end{equation}
    The posterior remains the same as under the unperturbed branch $i$. By induction over $t$, the full distribution over belief trajectories, and hence the future bases chosen along every trajectory, remains unchanged; points 3 and 4 are hence proven.
\end{proof}

\begin{lemma}[Tagging forces injectivity]
\label{lemma:tagging_injectivity}
For a tagged policy, the supported map
\begin{equation}
    (\gbm\eta,\tilde w_{t,\gbm\eta,i}) \mapsto (\tau(\gbm\eta,i), \tilde w_{t,\gbm\eta,i})
\end{equation}
is injective.
\end{lemma}
\begin{proof}
    Assume two outputs for the above map are identical. The equality of their battery components gives $\tilde w_{t,\gbm\eta,i} = \tilde w_{t,\gbm\eta',j}$. By construction, all tagged work values are distinct, so $(\gbm\eta,i) = (\gbm\eta',j)$. The corresponding inputs are therefore the same. Hence, this shows injectivity.
\end{proof}
Combining Lemmas~\ref{lemma:injective_unitary} and \ref{lemma:tagging_injectivity}, every tagged memory update can be carried out using an energy-conserving unitary. 
\subsection{Work penalty for tagging}
Having established the injectivity of the update map, we must quantify the average work deficit incurred by introducing this tagging perturbation. 

Recall that the expected work extracted from the expected state $\xi$ using a protocol tailored for $\rho$ is given by
\begin{equation}
    \beta\langle W\rangle_\rho = \D(\xi\|\gamma)-\D(\xi\|\rho)~.
\end{equation}
Therefore, if the target state changes from $\rho_p$ to $\rho_q$ by tagging the eigenvalues $p_i\to q_i$, the difference in average work would be
\begin{equation}
\begin{split}
    \beta\langle W\rangle_{\rho_p} -\beta\langle W\rangle _{\rho_q} = \D(\xi\|\rho_q)-\D(\xi\|\rho_p)
\end{split}
\end{equation}
Since both $\rho_p$ and $\rho_q$ share the same eigenbasis, this expression can be further simplified to the classical Kullback-Leibler divergence:
\begin{equation}
    \begin{split}
        \beta\langle W\rangle_{\rho_p} -\beta\langle W\rangle _{\rho_q} &= \sum_{i=1}^d p_i\ln\frac{p_i}{q_i}\\
        &=\D(p\|q)~,
    \end{split}
\end{equation}
where $p_i=\bra{\psi_i}\xi\ket{\psi_i}$. From Eq.~\eqref{eq:perturb_spec}, we know $q_i \geq (1-\alpha)p_i$, and hence for every $p_i>0$,
\begin{equation}
    \ln\frac{p_i}{q_i} \leq -\ln(1-\alpha)~.
\end{equation}
Therefore the extra dissipation is given by 
\begin{equation}
    \beta\langle W\rangle_{\rho_p} -\beta\langle W\rangle _{\rho_q} \leq -[\ln(1-\alpha)]~.
\end{equation}
We can now formally state the primary operational theorem.
\begin{theorem}
Consider a finite-horizon policy defined over a finite set of belief states and finite system dimensions, operating with a degenerate memory and retained ideal battery records. For any desired total work tolerance $\Delta>0$, there exists a tagged policy such that:
\begin{enumerate}
    \item It uses the same measurement basis as the original policy at all time steps and beliefs.
    \item The induced branch probabilities, Bayesian beliefs, and future basis choices match the untagged policy perfectly.
    \item Every memory update can be implemented via a unitary on the memory and current battery at zero additional thermodynamic cost.
    \item The memory can be returned to its initial state by sequentially applying the inverse update unitaries in reverse temporal order.
    \item Its cumulative expected work extraction is at most $\Delta$ below the theoretical optimum of the original policy.
\end{enumerate}
\end{theorem}
\begin{proof}
    We choose 
    \begin{equation}
        \alpha = 1-\exp\frac{-\beta\Delta}{L}~.
    \end{equation}
    At every time step and belief, construct the tagged spectrum using Lemma~\ref{lemma:non-colliding}. Points 1 and 2 follow from Lemma~\ref{lemma:spectral_preserve}, while point 3 follows from Lemmas~\ref{lemma:injective_unitary} and~\ref{lemma:tagging_injectivity}. Point 4 follows from the unitarity of the update maps and the retained battery records, which allow the inverse updates to be applied in reverse temporal order.

    Finally, the total penalty in the average cumulative work extraction over $L$ steps is bounded by:
    \begin{equation}
        \mathcal{F}_{\text{TO}}^{(1:L)}-\E_{\mathbf{H}_{L}|\Lambda_{\text{tag}}}\sum_{t=1}^L(W_t) \leq Lk_BT[-\ln(1-\alpha)]=Lk_BT\frac{\beta\Delta}{L}= \Delta~,
    \end{equation}
    proving point 5.
\end{proof}

\subsection{Operational Justification}
\label{sec:operational}
To clarify how these sequential updates can be executed and eventually uncomputed without incurring a Landauer erasure cost, we construct an explicit operational model.

We first cast the $N$ discrete belief states as $N$ mutually orthogonal quantum states $\{\ket{i}_M\}_{i=1}^N$, each of which would be mapped to different points on the probability simplex. Crucially, the memory of the agent is assumed to be degenerate, i.e., its Hamiltonian is $\mathcal{H}_M=\mathds{1}_M$. At the beginning of the operation, the agent’s memory is initialized in a state. Here, without loss of generality, we assume that the agent's memory is initialized in the state $\ket{K_0}_M$. After time step $t$, the battery system would take the form of $\rho_{B_t} = \sum_{i}p_i\ket{w_i}\!\bra{w_i}$, which is then projected onto $\ket{w_i}\!\bra{w_i}$ after the measurement. The memory is then updated unitarily conditioned on the battery via a controlled operation in the form of
\begin{equation}
\label{eq:update unitary}
    V^{(t)}_{BM} = \sum _i \ket{w_i}\!\bra{w_i}_{B_t}\otimes U^{(i)}_M~,
\end{equation}
where $U^{(i)}$ updates the memory from $\ket{K_{t-1}}$ to $\ket{K_t}$, conditioned on the work value $w^{(i)}$ being observed, according to Bayesian inference in Eq.~\eqref{eq:bayesian_update}. Since the Hamiltonian of the memory is degenerate, this update does not require energy.

At the end of $L$ time steps, the joint state of the agent's memory and the batteries is given by 
\begin{equation}
\Psi_{final} = \sum_{\mathbf{i}}\Pr(\mathbf{i})\left(\bigotimes_{t=1}^L\ket{\tilde w(\mathbf{i})}\!\bra{\tilde w(\mathbf{i})}_{B_t}\right)\otimes \ket{K_L(\mathbf{i})}\!\bra{K_L(\mathbf{i})}_M~,
\end{equation}
where $\mathbf{i}=(i_1,\cdots, i_L)$ labels a supported trajectory.
 To ensure the cost-free reset, the agent does not have to keep track of all past work values; instead, it can simply couple its memory system $M$ back to the battery array $\{B_t\}_{t=1}^L$. It can then apply the inverse controlled unitaries $V_{BM}^{(t)\dagger}$ in reverse temporal order, going backwards from $t=L$ to $t=1$. 

For example, at $t=L$, the agent couples its memory in the state of $\ket{K_L}$ with battery $B_L$ and applies $V^{(L)}_{BM}$. Notice that the state of the battery will not change because the control operation acts on its eigenbasis, while the state of the memory will be changed to $\ket{K_{L-1}}_M$. By iteratively coupling the memory register $M$ and the battery $B_L, B_{L-1},\cdots B_1$, the agent can deterministically drive its memory backward along its trajectory until it resets to $\ket{K_0}_M$. Crucially, this operation does not require the agent to keep track of all past work values, since the agent already has access to the sequence of batteries, which encodes this information.


\end{document}